\documentclass[11pt,article,reqno]{amsart}
\usepackage{amssymb}
\usepackage{amsthm}
\usepackage{amsmath}
\usepackage{graphicx}
\usepackage{subcaption}
\usepackage{hyperref}
\usepackage{color}

\usepackage{geometry}
\numberwithin{equation}{section}

\newtheorem{theorem}{Theorem}[section]
\newtheorem{corollary}[theorem]{Corollary}
\newtheorem{definition}[theorem]{Definition}
\newtheorem{lemma}[theorem]{Lemma}
\newtheorem{problem}[theorem]{Problem}

\def\R{\mathbb{R}}

\def\C{\mathbb{C}}

\def\T{\mathbb{T}}
\def\Sa{S^{\langle\alpha\rangle}}
\def\Saa{S^{\langle\!\langle \alpha\rangle\!\rangle}}
\def\Sza{S_0^{\langle\alpha\rangle}}
\def\Szaa{S_0^{\langle\!\langle \alpha\rangle\!\rangle}}
\def\Ta{T^{\langle\alpha\rangle}}

\begin{document}

\title[Maps that preserve a fixed quantum angle]{The structure of maps on the space of all quantum pure states that preserve a fixed quantum angle}
\author{Gy\"orgy P\'al Geh\'er}
\address{Gy\"orgy P\'al Geh\'er, Department of Mathematics and Statistics\\ University of Reading\\ Whiteknights\\ P.O.
Box 220\\ Reading RG6 6AX\\ United Kingdom}
\email{G.P.Geher@reading.ac.uk or gehergyuri@gmail.com}
\author{Michiya Mori}
\address{Michiya Mori, Graduate School of Mathematical Sciences, The University of Tokyo, Komaba, Tokyo, 153-8914, Japan}
\email{mmori@ms.u-tokyo.ac.jp}

\thanks{This research was supported by the London Mathematical Society Grant, Research in Pairs (Scheme 4), Reference No.: 41864}
\thanks{Geh\'er was supported by the Leverhulme Trust Early Career Fellowship (ECF-2018-125), and also by the Hungarian National Research, Development and Innovation Office (Grant no. K115383 and K134944)}
\thanks{Mori was supported by Leading Graduate Course for Frontiers of Mathematical Sciences and Physics (FMSP) and JSPS Research Fellowship for Young Scientists (KAKENHI Grant Number 19J14689), MEXT, Japan.}


\begin{abstract}
	Let $H$ be a Hilbert space and $P(H)$ be the projective space of all quantum pure states. Wigner's theorem states that every bijection $\phi\colon P(H)\to P(H)$ that preserves the quantum angle between pure states is automatically induced by either a unitary or an antiunitary operator $U\colon H\to H$. Uhlhorn's theorem generalises this result for bijective maps $\phi$ that are only assumed to preserve the quantum angle $\frac{\pi}{2}$ (orthogonality) in both directions. Recently, two papers, written by Li--Plevnik--\v{S}emrl and Geh\'er, solved the corresponding structural problem for bijections that preserve only one fixed quantum angle $\alpha$ in both directions, provided that $0 < \alpha \leq \frac{\pi}{4}$ holds. In this paper we solve the remaining structural problem for quantum angles $\alpha$ that satisfy $\frac{\pi}{4} < \alpha < \frac{\pi}{2}$, hence complete a programme started by Uhlhorn. In particular, it turns out that these maps are always induced by unitary or antiunitary operators, however, our assumption is much weaker than Wigner's.
\end{abstract}

\maketitle


\section{Introduction}\label{sec:1}

Let $H$ be a complex Hilbert space. For any vector $v\in H$ with length one, $\|v\| = 1$, let $[v]$ denote the line (one-dimensional subspace) it generates: $\C\cdot v$. From now on whenever we write $[v]$ with $v\in H$, it is implicitly assumed that $\|v\|=1$ holds. Also, given a finite number of vectors $v_1, v_2, \dots, v_n \in H$ with $\|v_1\|=\|v_2\|=\dots=\|v_n\| = 1$, the symbol $[v_1, v_2, \dots, v_n]$ stands for the subspace generated by them.
The \emph{projective space} $P(H)$ is the set of all lines in $H$, that is, $P(H) = \{[v]\colon v\in H, \|v\|=1\}$. In the mathematical foundations of quantum mechanics a line $[v]$ corresponds to a quantum pure state, and $P(H)$ to the set of all quantum pure states in a quantum system.
The so-called \emph{quantum angle} or \emph{Fubini--Study distance} between two lines $[u], [v]\in P(H)$ is defined by the following formula:
$$
\measuredangle([u],[v]) := \arccos {|\langle u,v\rangle|} \in \left[0,\frac{\pi}{2}\right].
$$
It is well-known that this is a metric on $P(H)$. Moreover, the important quantity called transition probability between $[u]$ and $[v]$ can be expressed as $\cos^2\measuredangle([u],[v])$, for more details on this see for instance the introduction of \cite{G}. 

Let us introduce the notation $\T := \{z\in\C\colon |z|=1\}$ for the complex unit circle.
In 1931 Wigner stated the following theorem.

\begin{theorem}[Wigner, \cite{W}]\label{thm:W}
Let $H$ be a complex Hilbert space with $\dim H \geq 2$. Assume that the bijective map $\phi\colon P(H)\to P(H)$ preserves the quantum angle between lines, that is,
\begin{equation}\label{eq:assW}
\measuredangle(\phi([u]),\phi([v])) = \measuredangle([u],[v]) \qquad ([u],[v] \in P(H)).
\end{equation}
Then $\phi$ is induced by either a unitary or an antiunitary operator $U\colon H\to H$, namely, we have
\begin{equation}\label{eq:Ws}
\phi([v]) = [Uv] \quad ([v]\in P(H)).
\end{equation}
Moreover, two unitary or antiunitary operators $U_1$ and $U_2$ induce the same map on $P(H)$ if and only if $U_2 = \lambda U_1$ holds with some $\lambda\in\T$. 
\end{theorem}

We note that the reverse direction is trivially true, namely, if $\phi$ has the form \eqref{eq:Ws}, then $\phi$ is clearly bijective and \eqref{eq:assW} holds. The real achievement here is that assuming only \eqref{eq:assW} and bijectivity already implies the remarkably regular structure \eqref{eq:Ws}.
We call a map a \emph{Wigner symmetry} if it possesses the form \eqref{eq:Ws}.
The above theorem became a cornerstone of the mathematical foundations of quantum mechanics. One reason being that it plays a crucial role in obtaining the general time-dependent Schr\"odinger equation through purely mathematical means. For a nice exposition regarding this we suggest Simon's paper \cite{Simon}. 

We note that Wigner himself did not give a mathematically rigorous proof of his statement, indeed, the proof presented in \cite{W} contains gaps. Interestingly enough, it took thirty years for the first mathematically rigorous proofs to appear, see \cite{Ba, LM, U}. In particular, in \cite{U} Uhlhorn proved a more general version of the above theorem for Hilbert spaces of dimension at least three. Namely, he only assumed the preservation of the quantum logical structure, while Theorem \ref{thm:W} assumes that its complete probabilistic structure is preserved. Still, Uhlhorn's conclusion is the same as Wigner's, which is a quite remarkable phenomenon.

\begin{theorem}[Uhlhorn, \cite{U}]\label{thm:U}
Let $H$ be a complex Hilbert space with $\dim H \geq 3$ and $\phi\colon P(H)\to P(H)$ be a bijective map preserving orthogonality in both directions, that is,
$$
\measuredangle(\phi([u]),\phi([v])) = \frac{\pi}{2} \;\;\iff\;\; \measuredangle([u],[v]) = \frac{\pi}{2} \qquad ([u],[v] \in P(H)).
$$
Then $\phi$ is a Wigner symmetry. Namely, there exists either a unitary or an antiunitary operator $U\colon H\to H$ such that
$$
\phi([v]) = [Uv] \quad ([v]\in P(H)).
$$
\end{theorem}

We note that Uhlhorn's theorem obviously fails to be true in a two-dimensional Hilbert space, since in that case for every line there exists only one line orthogonal to it. 
The above two theorems have been generalised in many ways, more on this can be found in the introduction of \cite{G}.

In this paper we are interested in the following problem which proposes to generalise Wigner's theorem along the direction of Uhlhorn.
\begin{problem}\label{prob}
Fix a quantum angle $0 < \alpha < \frac{\pi}{2}$. Can we characterise all bijective mappings $\phi\colon P(H)\to P(H)$ that preserve the quantum angle $\alpha$, that is,
$$
\measuredangle(\phi([u]),\phi([v])) = \alpha \;\;\iff\;\; \measuredangle([u],[v]) = \alpha \qquad ([u],[v] \in P(H))?
$$
\end{problem}
We emphasise that, like in Uhlhorn's theorem, nothing is assumed a priori about other angles, hence $\measuredangle(\phi([u]),\phi([v])) \neq \measuredangle([u],[v])$ might happen if $\measuredangle([u],[v]) \neq \alpha$.
Recently, the papers \cite{G,LPS} solved this problem for real Hilbert spaces.
However, for \emph{complex} Hilbert spaces it was only partially solved, we state the two relevant theorems below. The first one is the complete solution for two-dimensional Hilbert spaces.

\begin{theorem}[Geh\'er, \cite{G}]\label{thm:2d}
	Let $H$ be a complex Hilbert space with $\dim H = 2$ and fix a number $0 < \alpha < \frac{\pi}{2}$. Assume that $\phi\colon P(H) \to P(H)$ is a bijective map preserving the quantum angle $\alpha$ in both directions, that is,
	$$
	\measuredangle([u],[v]) = \alpha \;\iff\; \measuredangle(\phi([u]),\phi([v])) = \alpha \quad ([u],[v] \in P(H)).
	$$
	Then 
	\begin{itemize}
		\item[(i)] either $\phi$ is a Wigner symmetry,
		\item[(ii)] or $\alpha = \frac{\pi}{4}$, and there exists a Wigner symmetry $\psi$ such that
		\begin{equation}\label{eq:pworth}
			\phi([v]) \in \left\{\psi([v]), \psi([v])^\perp\right\} \qquad ([v] \in P(H)),
		\end{equation}
		where $\psi([v])^\perp$ denotes the unique line which is orthogonal to $\psi([v])$.
		Moreover, every bijective map $\phi$ that satisfies \eqref{eq:pworth} preserves the angle $\frac{\pi}{4}$.
	\end{itemize}
\end{theorem}

Theorem \ref{thm:2d} can be proved using the famous Bloch representation and a characterisation of bijective maps on the unit sphere of a real Hilbert space that preserve a fixed spherical angle (see \cite[Theorem 2.1]{G}).
The next theorem is the solution for quantum angles at most $\frac{\pi}{4}$.

\begin{theorem}[Geh\'er, \cite{G}]\label{thm:gen}
	Let $H$ be a complex Hilbert space with $\dim H \geq 3$ and fix a number $0<\alpha\leq\frac{\pi}{4}$.
	Assume that $\phi\colon P(H) \to P(H)$ is a bijective map which satisfies
	\begin{equation*}
	\measuredangle([u],[v]) = \alpha \;\iff\; \measuredangle(\phi([u]),\phi([v])) = \alpha \quad ([u],[v] \in P(H)).
	\end{equation*}
	Then $\phi$ is a Wigner symmetry.
\end{theorem}

In the present paper our goal is to solve Problem \ref{prob} for the remaining case when $\dim H \geq 3$ and $\frac{\pi}{4}<\alpha<\frac{\pi}{2}$.
Before we state our main theorem, let us briefly explain the strategy used in \cite{G} to prove Theorem \ref{thm:gen}. For a subset $S\subset P(H)$, we define its \emph{$\alpha$-set} by
$$
\Sa := \left\{[v]\in P(H)\colon \measuredangle ([v], [u]) = \alpha \text{ for all }[u]\in S\right\},
$$ 
and its \emph{double-$\alpha$-set} by
$$
\Saa := \left(\Sa\right)^{\langle\alpha\rangle}.
$$ 
The core idea of \cite{G} is to examine the $\alpha$-sets of pairs of lines.
More precisely, it turns out that if $0<\alpha<\frac{\pi}{4}$, then the set $\{[v_1],[v_2]\}^{\langle\alpha\rangle}$ contains exactly one pair of elements $[w_1],[w_2]$ with $\measuredangle([w_1],[w_2]) = \alpha$ if and only if $\measuredangle([v_1],[v_2]) = \beta$, where $\beta$ is explicitly given in terms of $\alpha$. Hence the angle $\beta$ is also preserved by $\phi$. Using this observation it is then possible to construct a sequence of quantum angles $\{\beta_n\}_{n=1}^\infty \subset \left(0,\frac{\pi}{2}\right)$ which are all preserved by $\phi$, moreover, $\beta_n\searrow 0$ as $n\to\infty$. Since small angles are preserved, one can prove that all angles must be preserved. For the case $\alpha = \frac{\pi}{4}$ a somewhat modified idea can be applied, which we do not detail here.

As was pointed out in \cite{G}, the above idea fails to work for quantum angles $\alpha > \frac{\pi}{4}$. The main result of this paper is to show that nonetheless the conclusion of Theorem \ref{thm:gen} holds for all quantum angles.

\begin{theorem}\label{thm:main}
	Let $H$ be a complex Hilbert space with $\dim H \geq 3$ and fix a number $\frac{\pi}{4}<\alpha<\frac{\pi}{2}$.
	Assume that $\phi\colon P(H) \to P(H)$ is a bijective map which preserves the quantum angle $\alpha$ in both directions, namely, it satisfies
	\begin{equation*}\label{eq:ass}
	\measuredangle([u],[v]) = \alpha \;\iff\; \measuredangle(\phi([u]),\phi([v])) = \alpha \quad ([u],[v] \in P(H)).
	\end{equation*}
	Then $\phi$ is a Wigner symmetry, that is, there exists a unitary or an antiunitary operator $U\colon H\to H$ such that
	$$
	\phi([v]) = [Uv] \quad ([v]\in P(H)).
	$$
\end{theorem}

We say that three lines $[v_1],[v_2],[v_3]$ are collinear if $\dim [v_1,v_2,v_3] \leq 2$. For any (closed) subspace $M\subset H$ we may identify the projective space $P(M)$ with the subset $\{[v]\in P(H)\colon v\in M, \|v\|=1\}\subset P(H)$.
If $\dim M = 2$, then we call $P(M)$ $(\subset P(H))$ a projective line.
The following definition plays a central role in our considerations.

\begin{definition}[Highly-$\alpha$-symmetric set]\label{defi}
A subset $T\subset P(H)$ is called highly-$\alpha$-symmetric if it satisfies the following three conditions: 
\begin{itemize}
\item[(i)] $\# T =\infty$,
\item[(ii)] $\# \Ta = \infty$,
\item[(iii)] for any subset $S\subset T$ with $\# S = 3$, $\Saa = T$.
\end{itemize}
\end{definition}

We now briefly explain our strategy to prove the above theorem.
The aim of the next section is to explore the structure of the $\alpha$-sets of three collinear lines, and to prove some auxiliary results. Then in sections \ref{sec:3} and \ref{sec:4} we investigate how highly-$\alpha$-symmetric sets look like when $\dim H \geq 4$ and $\dim H = 3$, respectively. It turns out that if $H$ has dimension at least four, then a set $T$ is highly-$\alpha$-symmetric if and only if it is a subset of a projective line with an additional special structure, described in Definition \ref{def:circle}. In case when the dimension of the Hilbert space is three, the aforementioned implication holds only in one direction.
In contrast with \cite{G} where $\alpha$-sets of pairs of lines were examined, here the core of our method is to \emph{explore the shape of double-$\alpha$-sets of general triples of lines}.
Using these insights we then prove in Section \ref{sec:proof} that all maps $\phi$ which satisfy our conditions necessarily map projective lines onto projective lines.
Finally, an application of Theorem \ref{thm:2d} will complete the proof.


\section{Some preliminary results}\label{sec:2}
From now on $H$ denotes a complex Hilbert space with $\dim H\geq 3$, and $\alpha$ is a fixed angle with $\frac{\pi}{4}<\alpha < \frac{\pi}{2}$. 
We begin with a lemma about some basic properties of $\alpha$-sets.

\begin{lemma}\label{basic}
We have the following relations: 
\begin{itemize}
	\item[(i)] If $S\subset P(H)$, then $S\subset \Saa$.
	\item[(ii)] If $S_1\subset S_2\subset P(H)$, then $S_1^{\langle\alpha\rangle} \supset S_2^{\langle\alpha\rangle}$ and $S_1^{\langle\!\langle \alpha\rangle\!\rangle}\subset S_2^{\langle\!\langle \alpha\rangle\!\rangle}$.
	\item[(iii)] If $S\subset P(H)$, then $(\Sa)^{\langle\!\langle \alpha\rangle\!\rangle} = \Sa$.
	\item[(iv)] Every highly-$\alpha$-symmetric set $T$ satisfies
	\begin{equation*}
		\Saa = T, \;\;\; \Sa = \Ta \qquad (S\subset T, \# S \geq 3).
	\end{equation*}
\end{itemize}

\end{lemma}

\begin{proof}
	Points (i)--(ii) are trivial by definition. Point (iii) is an easy application of (i)--(ii), and part (iv) is straightforward from (i)--(iii).
\end{proof}

As usual, we say two lines $[u], [v]\in P(H)$ are orthogonal if $\measuredangle([u],[v]) = \frac{\pi}{2}$. We introduce the notation $\perp$ for the orthogonality of vectors and subsets in $H$, and also for the orthogonality of lines in $P(H)$.
We continue with two lemmas about the general form of a pair of lines and its $\alpha$-set.

\begin{lemma}\label{2}
Let $[v_1], [v_2]\in P(H)$ be two different lines. 
Then there exist an orthonormal system $\{e_1, e_2\}\subset H$ and real numbers $c\geq d>0$, $c^2+d^2 = 1$ such that 
\begin{equation*}
[v_1] = [ce_1+ ide_2], \;\;\; [v_2] = [ce_1- ide_2].
\end{equation*}
\end{lemma}

\begin{proof}
	An application of the famous Bloch representation gives a simple proof. 
	However, in case the reader is not that familiar with it, a more direct proof can be given as follows. Since $[v_j] = [\lambda v_j]$ for all $\lambda\in\T$ and $j=1,2$, without loss of generality we may assume that $\langle v_1, v_2\rangle \geq 0$. Hence $v_1+v_2 \perp v_1-v_2$ and $0 < \|v_1-v_2\| \leq \|v_1+v_2\|$ hold. Since $\|v_1+v_2\|^2 + \|v_1-v_2\|^2 = 4$, there exist two numbers $c\geq d>0$, $c^2+d^2 = 1$ and an orthonormal system $\{e_1, e_2\}$ such that $v_1+v_2 = 2c e_1$ and $v_1-v_2 = 2id e_2$. From here a calculation gives the desired form.
\end{proof}

We introduce the notation $\sqcup$ for the disjoint union. We also set $a := \cos\alpha$ which we shall use throughout the paper.

\begin{lemma}\label{shape}
Let $\{e_1, e_2\}\subset H$ be an orthonormal system and $c\geq d>0$ with $c^2+d^2 = 1$. 
Define the function
$$
\rho\colon [-\theta_0,\theta_0]\to [0,1], \;\;\; \rho(\theta) = \sqrt{1-\left(\frac{a}{c}\right)^2\cos^2\theta-\left(\frac{a}{d}\right)^2\sin^2\theta},
$$
where 
\begin{itemize}
	\item if $a\leq d$, then $\theta_0 = \frac{\pi}{2}$,
	\item if $a > d$, then $\theta_0$ is the unique number with $0<\theta_0< \frac{\pi}{2}$ and $\left(\frac{a}{c}\right)^2\cos^2\theta_0+\left(\frac{a}{d}\right)^2\sin^2\theta_0 = 1$.
\end{itemize}
Then we have
\begin{equation}\label{eq:shape}
\{[ce_1+ ide_2], [ce_1- ide_2]\}^{\langle \alpha\rangle} = \bigsqcup\left\{ \mathcal{A}_{\theta} \colon -\theta_0 \leq \theta \leq \theta_0, \; \theta\neq-\frac{\pi}{2} \right\},
\end{equation}
where
\begin{equation}\label{eq:A-theta}
\mathcal{A}_{\theta}:= \left\{ \left[\frac{a}{c}\cos\theta \cdot e_1 + \frac{a}{d}\sin\theta \cdot e_2 + h\right] \colon
h\perp \{e_1, e_2\}, \|h\| = \rho(\theta) \right\}.
\end{equation}
\end{lemma}

\begin{proof}
Notice that by our assumptions we always have $c > a$.
Since $0 < \frac{a}{c} \leq \frac{a}{d}$, the function $\theta\mapsto \left(\frac{a}{c}\right)^2\cos^2\theta+\left(\frac{a}{d}\right)^2\sin^2\theta$ is positive-valued, monotone non-increasing on $[-\frac{\pi}{2},0]$, and monotone non-decreasing on $[0,\frac{\pi}{2}]$. As $\frac{a}{c} < 1$, we have a real number $0<\theta_0\leq \frac{\pi}{2}$ with the desired property.

Consider an arbitrary line $[v]\in P(H)$. 
We may take numbers $c_1\geq 0$, $c_2\in\C$ and a vector $h\perp \{e_1,e_2\}$ such that $c_1^2 + |c_2|^2 + \|h\|^2 = 1$ and $[v] = [c_1 e_1+c_2e_2+h]$.
Then we have $[v] \in \{[ce_1+ ide_2], [ce_1- ide_2]\}^{\langle \alpha\rangle}$ if and only if
$$
\left| c_1c + ic_2d \right| = \left| c_1c - ic_2d \right| = a.
$$
This is equivalent to
\begin{itemize}
\item either $c_1>0$, $c_2\in\R$ and $(c_1c)^2 + (c_2d)^2=a^2$,
\item or $c_1=0$ and $|c_2|d = a$, in which case we may assume without loss of generality that $c_2 = \frac{a}{d}$.
\end{itemize}
Therefore $c_1c = a\cos\theta$ and $c_2d = a\sin\theta$ for some $-\frac{\pi}{2}\leq\theta\leq\frac{\pi}{2}$, which proves the $\subseteq$ part of \eqref{eq:shape}. The $\supseteq$ part of \eqref{eq:shape} and the disjointness are obvious.
\end{proof}

Note that in case when $\theta_0 = \frac{\pi}{2}$, then the set $\mathcal{A}_{-\frac{\pi}{2}}$ is well defined by \eqref{eq:A-theta}, however, we have $\mathcal{A}_{-\frac{\pi}{2}} = \mathcal{A}_{\frac{\pi}{2}}$.
Throughout the paper whenever we use the symbols $c$ and $d$, it is always assumed that $c\geq d>0$ and $c^2+d^2 = 1$. Therefore, like in the above proof, the inequality $c>a$ is automatically satisfied.

Straightforward calculations give the following properties of $\rho$, which are also illustrated in Figure \ref{fig:rho} for the reader's convenience:
\begin{itemize}
	\item $\rho$ is an even continuous function on $[-\theta_0,\theta_0]$, differentiable on $(-\theta_0,\theta_0)$, and $\rho'(0) = 0$,
	\item if $d<\sqrt{\frac{1}{2}}$, then $\rho$ is strictly increasing on $[-\theta_0,0]$, and strictly decreasing on $[0,\theta_0]$,
	\item if $d=\sqrt{\frac{1}{2}}$, then $\rho$ is the constant $\sqrt{1-2a^2}$ function,
	\item $\rho(\theta_0) = 0$ if and only if $a\geq d$.
\end{itemize}

\begin{figure}[h]
    \flushleft
    \begin{subfigure}[t]{0.425\textwidth}
        \includegraphics[scale=0.5]{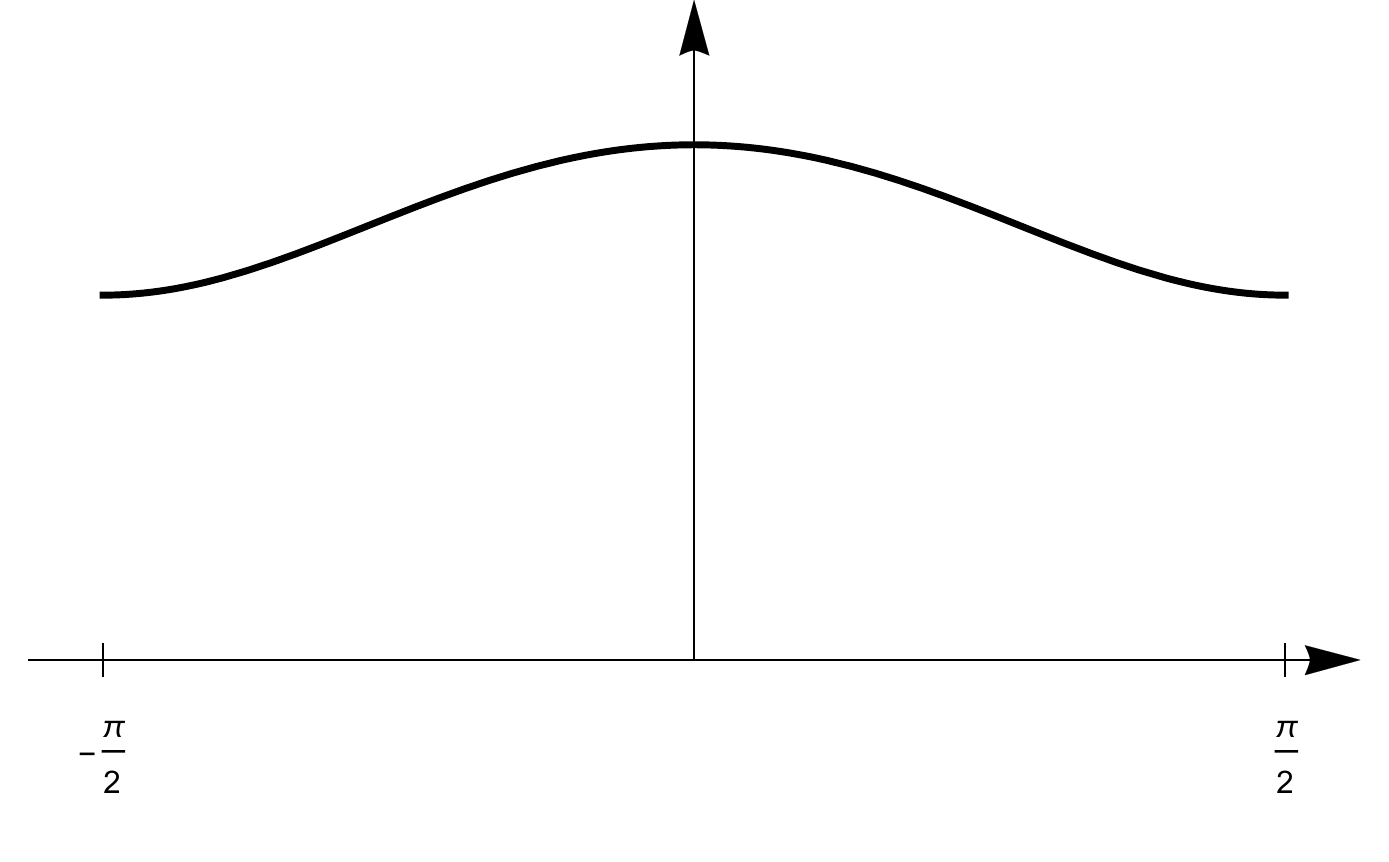}
        \caption{When $a<d<\sqrt{\frac{1}{2}}$. Then $\theta_0 = \frac{\pi}{2}$ and $\rho(\theta_0) > 0.$}
    \end{subfigure}
    \hspace{1cm}
    \begin{subfigure}[t]{0.425\textwidth}
        \includegraphics[scale=0.5]{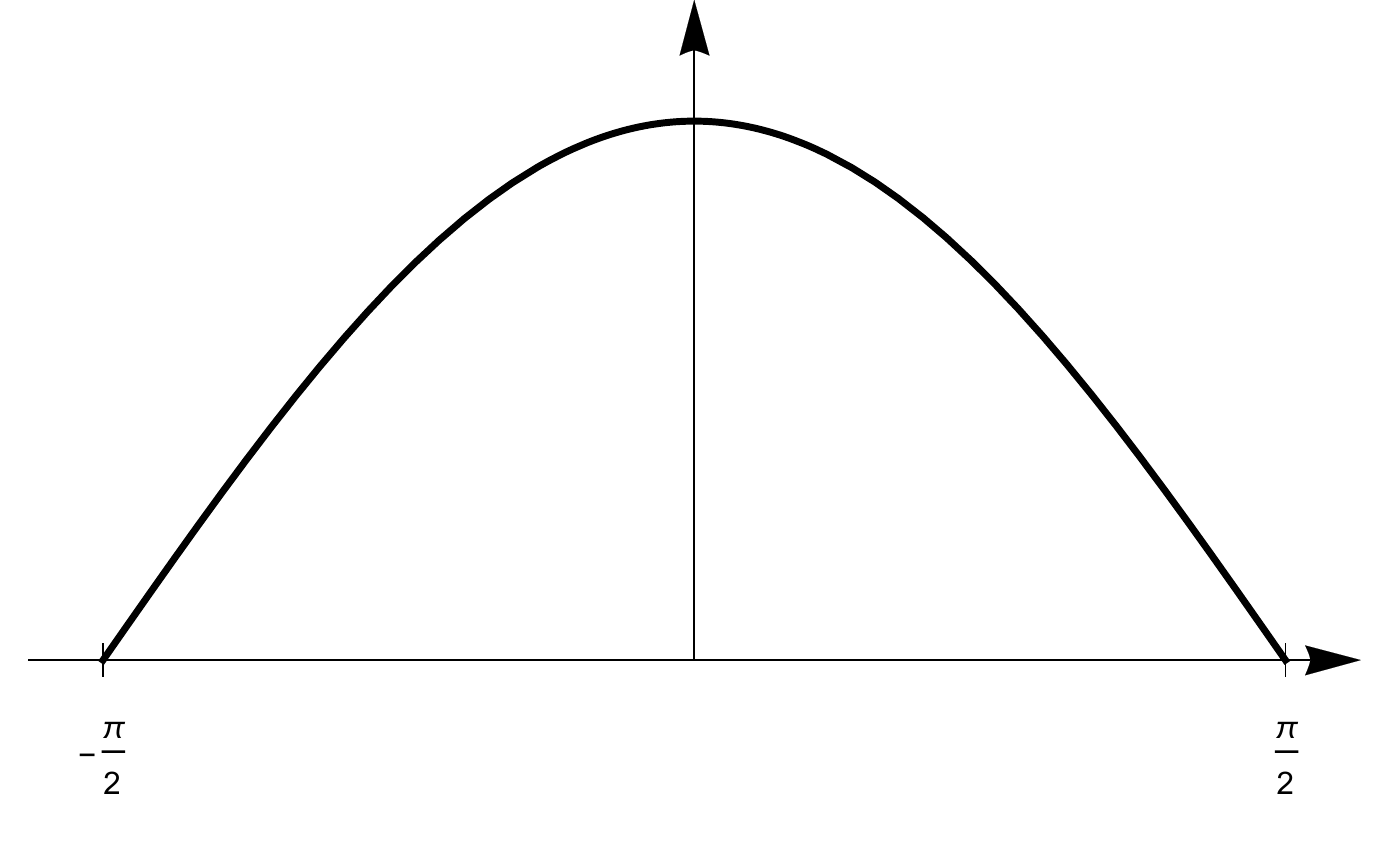}
        \caption{When $a=d<\sqrt{\frac{1}{2}}$. Then $\theta_0 = \frac{\pi}{2}$, and $\rho(\theta_0) = 0.$}
    \end{subfigure}
    
    \bigskip
    
    \begin{subfigure}[t]{0.425\textwidth}
        \includegraphics[scale=0.5]{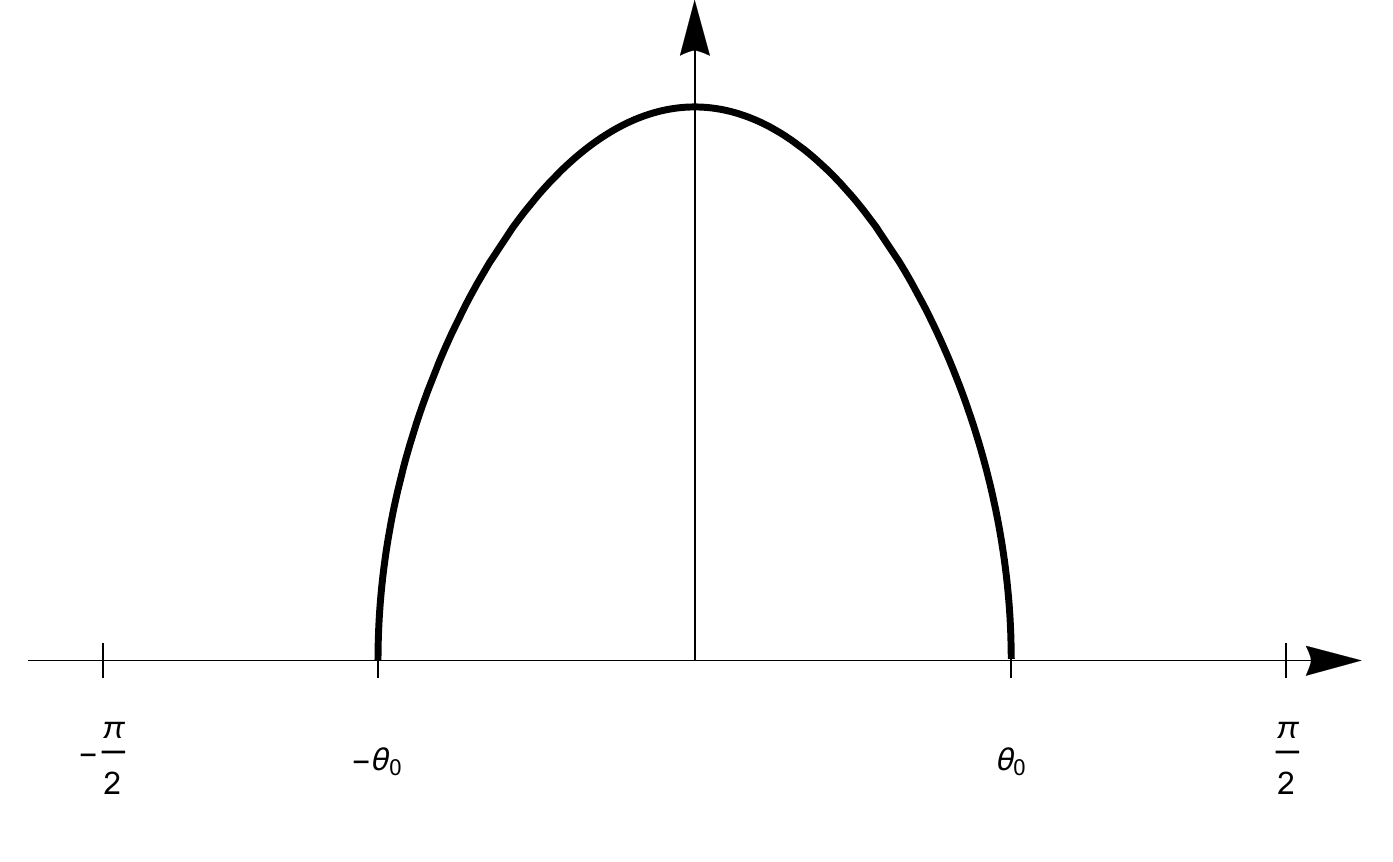}
        \caption{When $a>d$. Then $d<\sqrt{\frac{1}{2}}$, $0<\theta_0<\frac{\pi}
{2}$ and $\rho(\theta_0) = 0.$}
    \end{subfigure}
    \hspace{1cm}
    \begin{subfigure}[t]{0.425\textwidth}
        \includegraphics[scale=0.5]{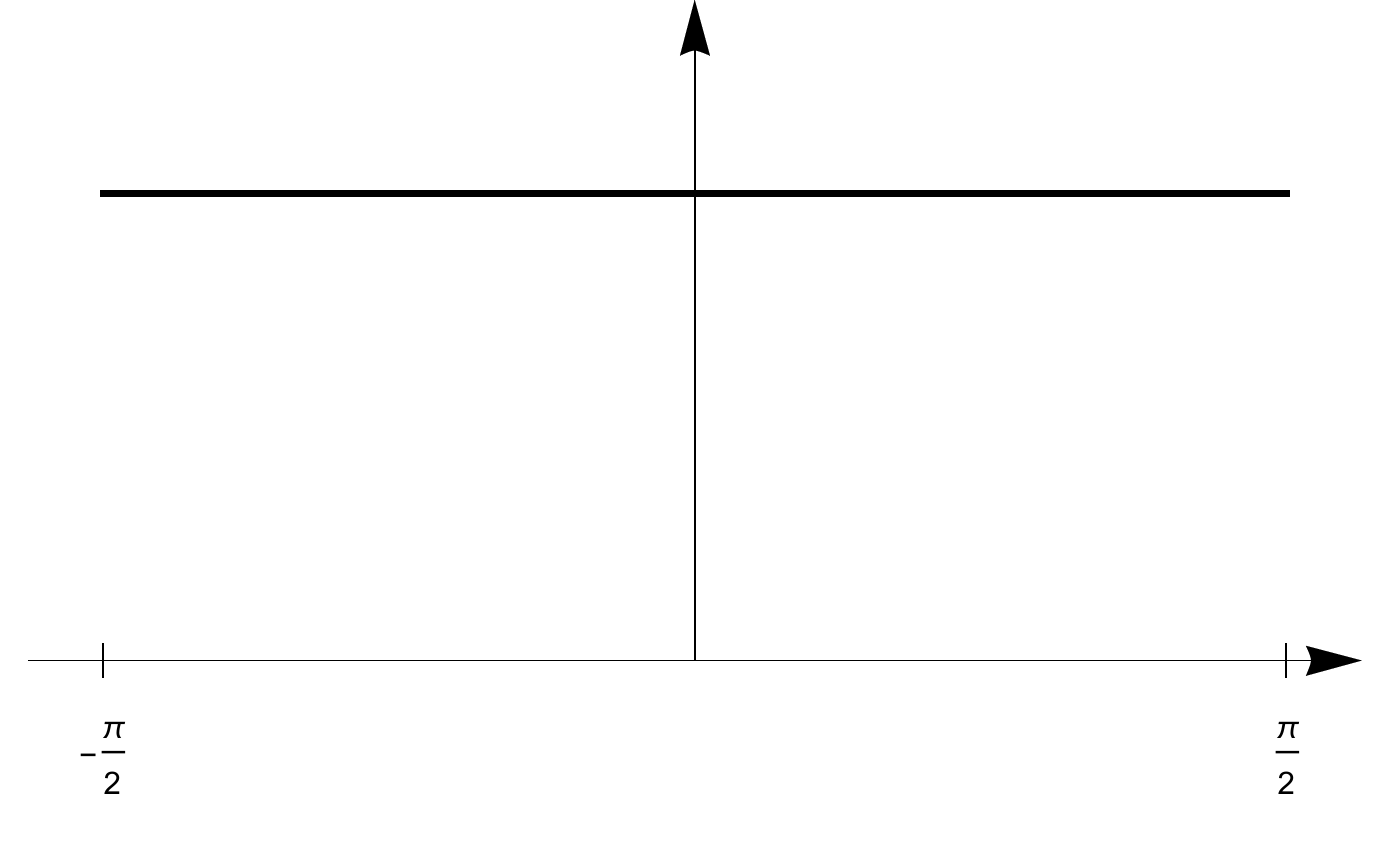}
        \caption{When $d=\sqrt{\frac{1}{2}}$. Then $\theta_0 = \frac{\pi}{2}$, and $\rho$ is a positive constant function.}
    \end{subfigure}
    \caption{Illustration of the function $\rho$.}\label{fig:rho}
\end{figure}

The following two lemmas give the general form of a collinear triple of lines and its $\alpha$-set.

\begin{lemma}\label{3}
Let $[v_1], [v_2], [v_3]\in P(H)$ be three collinear lines that are pairwise different.
Then there exist an orthonormal system $\{e_1, e_2\}\subset H$, three numbers $\lambda_1, \lambda_2, \lambda_3 \in \T$, and two real numbers $c\geq d>0$, $c^2+d^2 = 1$ such that   
$$
[v_j] = [ce_1+\lambda_{j}de_2] \qquad (j=1, 2, 3).
$$
\end{lemma}

\begin{proof}
	An application of the Bloch representation gives a geometric and simple proof. We give another more direct proof here. By Lemma \ref{2}, we can write $[v_1] = [\mathfrak{c}{f}_1+ i\mathfrak{d}{f}_2]$ and $[v_2] = [\mathfrak{c}{f}_1- i\mathfrak{d}{f}_2]$ where $\{{f}_1, {f}_2\}$ is an orthonormal system, $\mathfrak{c}\geq \mathfrak{d}>0$, $\mathfrak{c}^2+\mathfrak{d}^2 = 1$. A straightforward calculation gives that
	$$
	\left|\langle v_1, \cos t {f}_1 + \sin t {f}_2\rangle\right| = \left|\langle v_2, \cos t {f}_1 + \sin t {f}_2\rangle\right| \qquad \left(0\leq t\leq \frac{\pi}{2}\right).
	$$
	
	We may take numbers $\mathfrak{c}_1\geq 0$ and $\mathfrak{c}_2\in\C$ such that $\mathfrak{c}_1^2 + |\mathfrak{c}_2|^2 = 1$ and $[v_3] = [\mathfrak{c}_1{f}_1+\mathfrak{c}_2{f}_2]$. On the one hand, suppose that $\mathfrak{c}_1 \geq \mathfrak{c}$. Then $|\mathfrak{c}_2| \leq \mathfrak{d}$,
	$$
	\left|\langle v_1, {f}_1\rangle\right| = \left|\langle v_2, {f}_1\rangle\right| = \mathfrak{c} \leq \mathfrak{c}_1= \left|\langle v_3, {f}_1\rangle\right|
	$$
	and 
	$$
	\left|\langle v_1, {f}_2\rangle\right| = \left|\langle v_2, {f}_2\rangle\right| = \mathfrak{d} \geq |\mathfrak{c}_2|= \left|\langle v_3, {f}_2\rangle\right|.
	$$
	Therefore there exists a $0\leq t\leq \frac{\pi}{2}$ such that with $e_1 := \cos t {f}_1 + \sin t {f}_2$ we have
	\begin{equation}\label{eq:M1}
	\left|\langle v_1, e_1\rangle\right| = \left|\langle v_2, e_1\rangle\right| = \left|\langle v_3, e_1\rangle\right|.
	\end{equation}
	On the other hand, if $\mathfrak{c}_1 < \mathfrak{c}$, then we prove the existence of a line $[e_1]$ with \eqref{eq:M1} in a very similar way.
	
	Now, let $[e_2]$ be the unique line which is orthogonal to $[e_1]$ and is contained in the subspace $[v_1, v_2]$. Parseval's formula implies
	$$
	\left|\langle v_1, e_2\rangle\right| = \left|\langle v_2, e_2\rangle\right| = \left|\langle v_3, e_2\rangle\right|.
	$$
	By interchanging the role of $e_1$ and $e_2$ if necessary, we may assume $c := |\langle v_1, e_1\rangle| \geq |\langle v_1, e_2\rangle| =: d$, which completes the proof.
\end{proof}

\begin{lemma}\label{collin-alpha}
Let $c\geq d>0$ such that $c^2+d^2= 1$, $\lambda_1,\lambda_2,\lambda_3\in\T$ pairwise different, and $\{e_1, e_2\}$ an orthonormal system of $H$. Set $S_0:=\{[ce_1+\lambda_{j}de_2]\colon j=1, 2, 3\} \subset P(H)$. 
\begin{itemize}
	\item[(i)] If $a > d$, then 
	$$\Sza = \left\{\left[\frac{a}{c}e_1 + h\right] \colon h\in H,\, \lVert h\rVert = \sqrt{1-\frac{a^2}{c^2}},\, h\perp \{e_1, e_2\}\right\}.$$
	\item[(ii)] If $a \leq d$, then 
	\begin{align*}
	\Sza &= \left\{\left[\frac{a}{c}e_1 + h\right] \colon h\in H,\, \lVert h\rVert = \sqrt{1-\frac{a^2}{c^2}},\, h\perp \{e_1, e_2\}\right\} \\
	&\hspace{1.5cm} \bigsqcup \left\{\left[\frac{a}{d}e_2 + h\right] \colon h\in H,\, \lVert h\rVert = \sqrt{1-\frac{a^2}{d^2}},\, h\perp \{e_1, e_2\}\right\}.
	\end{align*}
\end{itemize}
\end{lemma}

\begin{proof}
	Note that $c > a$.
	Consider an arbitrary line $[v]\in P(H)$. 
	We may take numbers $c_1\geq 0$, $c_2\in\C$ and a vector $h\perp \{e_1,e_2\}$ such that $c_1^2 + |c_2|^2 + \|h\|^2 = 1$ and $[v] = [c_1 e_1+c_2e_2+h]$. Then we have $[v]\in \Sza$ if and only if
	$$
	\left|c_1c+c_2\overline{\lambda_j}d\right| = a \qquad (j=1,2,3).
	$$
	Since the numbers $\lambda_j$ are pairwise different, a simple geometric observation implies that
	$$
	|c_1c+c_2{\lambda}d| = a \qquad (\lambda\in\T).
	$$
	Thus $[v]\in \Sza$ if and only if
	\begin{itemize}
		\item either $c_2 = 0$, $c_1 = \frac{a}{c}$,
		\item or $c_1 = 0$, $|c_2| = \frac{a}{d}$.
	\end{itemize}
	Note that without loss of generality $c_2>0$ may be assumed in the latter case. This completes the proof.
\end{proof}

We finish this section with an important definition.

\begin{definition}[Circle]\label{def:circle}
For any orthonormal system $\{e_1,e_2\}\subset H$ and numbers $\mathfrak{c},\mathfrak{d}>0$, $\mathfrak{c}^2+\mathfrak{d}^2=1$, the set of the form $\{[\mathfrak{c}e_1 + \lambda \mathfrak{d}e_2] \colon \lambda\in \T\}$ is called a circle.
\end{definition}

Set $M := [e_1,e_2]$ with the above vectors and consider the Bloch representation of $P(M)$ (see for instance \cite{G}). Remark that a straightforward calculation shows that the image of the circle $\{[\mathfrak{c}e_1 + \lambda \mathfrak{d}e_2] \colon \lambda\in \T\}$ is an actual circle on the surface $\mathbb{S}^2$, hence the above choice of the name. Moreover, it is a great (or geodesic) circle if and only if $\mathfrak{c}=\mathfrak{d}=\frac{1}{\sqrt{2}}$.

In the forthcoming two sections we shall explore how the double-$\alpha$-set of $S_0$ looks like, and will also examine highly-$\alpha$-symmetric sets in detail.


\section{The structure of highly-$\alpha$-symmetric sets in the at least four-dimensional case}\label{sec:3}

Our goal in this section is to show that highly-$\alpha$-symmetric sets are exactly circles in $P(H)$ if $\dim H\geq 4$.
First, we calculate the double-$\alpha$-set of $S_0$ from Lemma \ref{collin-alpha}.

\begin{lemma}\label{circle4}
Using the notation and assumptions of Lemma \ref{collin-alpha}, suppose that $\dim H \geq 4$. Then we have
\begin{equation*}
\Szaa = \{[ce_1 + \lambda de_2] \colon \lambda\in \T\}.
\end{equation*}
\end{lemma}

\begin{proof}
	Recall that $c>a$. Define 
	\begin{equation*}
		\mathcal{C} := \left\{\left[\frac{a}{c}e_1 + h\right]\colon h\in H,\, \lVert h\rVert = \sqrt{1-\frac{a^2}{c^2}},\, h\perp \{e_1, e_2\}\right\}.
	\end{equation*}
	As $\mathcal{C} \subseteq \Sza$, we have $\mathcal{C}^{\langle\alpha\rangle} \supseteq \Szaa$.
	Consider a line $[v] = [c_1e_1+c_2e_2+k]$ with $c_1\geq 0$, $c_2\in\C$, $k\in H$, $k\perp\{e_1, e_2\}$, and $c_1^2+|c_2|^2+\|k\|^2 = 1$. We have $[v] \in \mathcal{C}^{\langle\alpha\rangle}$ if and only if 
	\begin{equation*}
	\left| c_1 \frac{a}{c} + \langle k, h\rangle \right| = a \qquad \left(h\in H,\, \lVert h\rVert = \sqrt{1-\frac{a^2}{c^2}},\, h\perp \{e_1, e_2\}\right).
	\end{equation*}
	Notice that the inner product $\langle k, h\rangle$ above runs through a closed disk of radius $\|k\|\cdot\sqrt{1-\frac{a^2}{c^2}}$ on the complex plane. As $c>a$, we obtain $k = 0$ and $c_1 = c$, hence
	$$ \mathcal{C}^{\langle\alpha\rangle} = \{[ce_1 + \lambda de_2] \colon \lambda\in \T\}. $$
	In case of (i) of Lemma \ref{collin-alpha}, this completes the proof. On the other hand, in case of (ii) of Lemma \ref{collin-alpha}, we easily see the reverse inclusion $\Szaa \supseteq \mathcal{C}^{\langle\alpha\rangle}$, hence the proof is done.
\end{proof}

Observe that Lemmas \ref{collin-alpha} and \ref{circle4} imply the following.

\begin{corollary}\label{cor:M}
	If $\dim H \geq 4$, then every circle in $P(H)$ is highly-$\alpha$-symmetric. 
\end{corollary}

For the remaining part of this section our aim is to prove the reverse. 

\begin{lemma}\label{nice}
	Assume that $\dim H \geq 4$. Then every highly-$\alpha$-symmetric set $T$ satisfies one of the following points:
	\begin{itemize}
	\item[(i)] either $T$ is contained in a projective line,
	\item[(ii)] or there exists a subspace $M$ with $\dim M = 3$ such that for all $[v_1],[v_2],[v_3] \in T$ pairwise different elements we have $\left[v_1,v_2,v_3\right] = M$.
	\end{itemize}
\end{lemma}

\begin{proof}
	Suppose that there exist $[u_1],[u_2],[u_3] \in T$ collinear and pairwise different. Then, by Lemmas \ref{3} and \ref{circle4}, the set $T = \{[u_1],[u_2],[u_3]\}^{\langle\!\langle \alpha\rangle\!\rangle}$ is a circle, hence (i) follows.
	
	From now on we assume otherwise.
	Consider three arbitrary pairwise different lines $[v_1],[v_2],[v_3] \in T$. Set $M := [v_1,v_2,v_3]$ which is a three-dimensional subspace. 
	Our goal is to prove $T\subset P(M)$, which will complete the proof. Note that
	\begin{align*}
	&\{[v_1],[v_2],[v_3]\}^{\langle\alpha\rangle} \\
	& \; = \left\{ [u+w] \colon u\in M, w\perp M, \|u\|^2+\|w\|^2=1, |\langle u, v_1 \rangle| = |\langle u, v_2 \rangle| = |\langle u, v_3 \rangle| = a \right\}.
	\end{align*}
	As this set is equal to $\Ta$, it is not empty.
	Let $[x+y] \in P(H)$ be an arbitrary line where $x\in M$, $y\perp M$ and $\|x\|^2+\|y\|^2 = 1$. 
	Clearly, we have $[x+y] \in T = \{[v_1],[v_2],[v_3]\}^{\langle\!\langle \alpha\rangle\!\rangle}$ if and only if 
	\begin{equation}\label{eq:xuyw}
	\left|\langle x,u\rangle + \langle y,w\rangle\right| = a
	\end{equation}
	holds for all $u\in M, w\perp M, \|u\|^2+\|w\|^2=1, |\langle u, v_1 \rangle| = |\langle u, v_2 \rangle| = |\langle u, v_3 \rangle| = a$.
	We point out that the only restriction on $w$ above, apart from being orthogonal to $M$, concerns its norm.
	Therefore, if $[x+y] \in T$ with $x\neq 0$, $y\neq 0$, then $T$ contains collinear triples, namely
	$$
	\left\{ [x + \lambda y] \colon \lambda\in\T \right\} \subset T,$$
	which is a contradiction.
	
	The above observations imply $T\subset P(M) \cup P(M^\perp)$, where $M^\perp$ denotes the largest subspace in $H$ orthogonal to $M$. 
	On the one hand, if $\dim H \geq 5$ and $[y] \in T \cap P(M^\perp)$, then \eqref{eq:xuyw} cannot hold. Hence in that case indeed $T\subset P(M)$ follows. 
	On the other hand, if $\dim H = 4$, then $T\subset P(M) \cup \{[e]\}$ where $e\perp \{v_1,v_2,v_3\}$, $\|e\|=1$. Assume for a moment that $[e]\in T$. Then a consideration of $\{[v_2],[v_3],[e]\}$ instead of $\{[v_1],[v_2],[v_3]\}$ gives that $T\subset P([v_2,v_3,e]) \cup \{[f]\}$ where $f\perp \{v_2,v_3,e\}$, $\|f\|=1$. Since $v_1 \notin [v_2,v_3,e]$, we have $[v_1] = [f]$. In such a way we eventually obtain that 
	$$T\subset \left(P([v_2,v_3,e]) \cup \{[v_1]\}\right) \cap \left(P([v_1,v_3,e]) \cup \{[v_2]\}\right) \cap \left(P([v_1,v_2,e]) \cup \{[v_3]\}\right) \cap \left(P(M) \cup \{[e]\}\right).$$ 
	Hence $\#T = 4$, a contradiction. Therefore, $[e]\notin T$, and we conclude $T\subset P(M)$.
\end{proof}

\begin{lemma}\label{infinite-element}
Assume that $\dim H \geq 4$. Let $\{e_1, e_2,e_3\}\subset H$ be an orthonormal system, $c\geq d>0$ with $c^2+d^2 = 1$, and $c_1, c_2\in\C$, $c_3 > 0$, $|c_1|^2+|c_2|^2+c_3^2 = 1$. Set $[v_1] = [ce_1+ ide_2]$, $[v_2] = [ce_1- ide_2]$, $[v_3] = [c_1e_1+c_2e_2+c_3e_3]\in P(H)$, and define the function
\begin{equation}\label{eq:z}
	z\colon [-\theta_0,\theta_0]\to\C, \;\;\; z(\theta) = c_1\frac{a}{c}\cos\theta+c_2\frac{a}{d}\sin\theta,
\end{equation}
where $\theta_0$ and $\rho$ are as in Lemma \ref{shape}.
Then for each $-\theta_0 \leq \theta \leq \theta_0$ we have $\#\left(\mathcal{A}_{\theta}\cap\{[v_3]\}^{\langle \alpha\rangle}\right) = \infty$ if and only if one of the following possibilities happens:
\begin{itemize}
	\item[(i)] either $|z(\theta)|-c_3\rho(\theta) < a < |z(\theta)|+c_3\rho(\theta)$,
	\item[(ii)] or $z(\theta) = 0$ and $\rho(\theta) = \frac{a}{c_3}$,
\end{itemize}
where $\mathcal{A}_{\theta}$ is as in \eqref{eq:A-theta}.

Moreover, we have $\#\left(\mathcal{A}_{\theta}\cap\{[v_3]\}^{\langle \alpha\rangle}\right) = 1$ if and only if
\begin{itemize}
	\item[(iii)]  $z(\theta) \neq 0$, and either $a = |z(\theta)|-c_3\rho(\theta)$, or $ a = |z(\theta)|+c_3\rho(\theta)$.
	\end{itemize}
\end{lemma}

\begin{proof}
	An element $\left[\frac{a}{c}\cos\theta \cdot e_1 + \frac{a}{d}\sin\theta \cdot e_2 + h\right]$ of $\mathcal{A}_{\theta}$ is in $\{[v_3]\}^{\langle \alpha\rangle}$ if and only if
	\begin{equation*}
	\left| z(\theta) + c_3\langle e_3, h\rangle \right| = a.
	\end{equation*}
	Notice that if we go through all elements of $\mathcal{A}_{\theta}$, then the complex number $c_3\langle e_3, h\rangle$ goes through a closed (possibly degenerate) disk of radius $c_3\rho(\theta)$. This radius is $0$ if and only if $\rho(\theta)=0$. 
	
	Assume that $z(\theta) \neq 0$. 
	Then by some elementary geometric observations we obtain the following possibilities:
	\begin{itemize}
	\item if $|z(\theta)|-c_3\rho(\theta) > a$ or $ a > |z(\theta)|+c_3\rho(\theta)$, 
	then $\mathcal{A}_{\theta}\cap\{[v_3]\}^{\langle \alpha\rangle} = \emptyset$,
	\item if $|z(\theta)|-c_3\rho(\theta) = a$ or $ a = |z(\theta)|+c_3\rho(\theta)$, 
	then $\#\left(\mathcal{A}_{\theta}\cap\{[v_3]\}^{\langle \alpha\rangle}\right) = 1$,
	\item if $|z(\theta)|-c_3\rho(\theta) < a < |z(\theta)|+c_3\rho(\theta)$,
	then $\#\left(\mathcal{A}_{\theta}\cap\{[v_3]\}^{\langle \alpha\rangle}\right) = \infty$.
	\end{itemize}
	
	In case when $z(\theta) = 0$, then we obtain the following possibilities: 
	\begin{itemize}
	\item if $c_3\rho(\theta) < a$, then $\mathcal{A}_{\theta}\cap\{[v_3]\}^{\langle \alpha\rangle} = \emptyset$,
	\item if $c_3\rho(\theta) \geq a$, then $\#\left(\mathcal{A}_{\theta}\cap\{[v_3]\}^{\langle \alpha\rangle}\right) = \infty$.
	\end{itemize}
	Notice that the case $c_3\rho(\theta) > a$ is included in (i) in the statement of the lemma.
\end{proof}

Notice that $\#\left(\mathcal{A}_{\theta}\cap\{[v_3]\}^{\langle \alpha\rangle}\right)$ is either $0$, or $1$, or $\infty$, provided that $\dim H \geq 4$.
Now, we are in the position to prove the main result of this section.

\begin{lemma}\label{circle-char-4}
	Assume that $\dim H \geq 4$. Then a set $T\subset P(H)$ is highly-$\alpha$-symmetric if and only if it is a circle.
\end{lemma}

\begin{proof}
	Corollary \ref{cor:M} gives one direction. 
	To prove the reverse implication, assume that $T$ is highly-$\alpha$-symmetric. By Lemma \ref{nice} and the assumption $\# T = \infty$, we may take a pair of different elements $[v_1], [v_2]\in T$ such that they are not orthogonal.
	Then, by Lemma \ref{2}, we have $[v_1] =[ce_1+ ide_2]$ and $[v_2] =[ce_1- ide_2]$ for some orthonormal system $\{e_1, e_2\}\subset H$ and real numbers $c > d > 0$, $c^2+d^2 = 1$. Define $\mathcal{A}_\theta$, $\rho$ and $\theta_0$ as in Lemma \ref{shape}. Consider an arbitrary third element $[u]\in T\setminus\{[v_1], [v_2]\}$. If $[u]$ sits on the projective line spanned by $[v_1]$ and $[v_2]$, then by Lemmas \ref{3} and \ref{circle4}, the set $T = \{[v_1], [v_2], [u]\}^{\langle\!\langle \alpha\rangle\!\rangle}$ is a circle.
	
	From now on we assume that $T\cap P([v_1, v_2]) = \{[v_1], [v_2]\}$. By Lemma \ref{nice}, there exists a unit vector $e_3\perp\{e_1,e_2\}$ such that $T\subset P([e_1,e_2,e_3])$. Consider two arbitrary (not necessarily different) lines $[v_3], [\widehat{v_3}] \in T\setminus\{[v_1], [v_2]\}$. We may take numbers $c_1, c_2\in\C$, $c_3> 0$, $|c_1|^2+|c_2|^2+c_3^2=1$, $\widehat{c_1}, \widehat{c_2}\in\C$, $\widehat{c_3}> 0$, $|\widehat{c_1}|^2+|\widehat{c_2}|^2+\widehat{c_3}^2=1$ such that 
	$$[v_3] = [c_1e_1+c_2e_2+c_3e_3] \;\;\; \text{and} \;\;\; [\widehat{v_3}] = [\widehat{c_1}e_1+\widehat{c_2}e_2+\widehat{c_3}e_3].$$
	By (iv) of Lemma \ref{basic}, we have 
	$$\{[v_1], [v_2],[v_3]\}^{\langle \alpha\rangle} = \Ta = \{[v_1], [v_2],[\widehat{v_3}]\}^{\langle \alpha\rangle}.$$ 
	By Lemma \ref{shape}, this implies 
	\begin{equation}\label{eq:A-th-cap}
	\mathcal{A}_\theta \cap \{[v_3]\}^{\langle \alpha\rangle} = \mathcal{A}_\theta \cap \{[\widehat{v_3}]\}^{\langle \alpha\rangle}
	\end{equation}
	for all $-\theta_0 \leq \theta \leq \theta_0$. We define the functions $z$ and $\widehat{z}$ by \eqref{eq:z} and
	\begin{equation*}
	\widehat{z}\colon [-\theta_0,\theta_0]\to\C, \;\;\; \widehat{z}(\theta) = \widehat{c_1}\frac{a}{c}\cos\theta+\widehat{c_2}\frac{a}{d}\sin\theta.
	\end{equation*}
	Clearly, \eqref{eq:A-th-cap} is equivalent to the following for all $-\theta_0 \leq \theta \leq \theta_0$:
	\begin{equation}\label{eq:A-theta-part}
	a = \left| z(\theta) + c_3\langle e_3,h\rangle \right| \iff a = \left| \widehat{z}(\theta) + \widehat{c_3}\langle e_3, h\rangle \right| \qquad (h\perp \{e_1, e_2\}, \|h\| = \rho(\theta)).
	\end{equation}
	
	Assume for a moment that $[e_3]\in T$. 
	Substitute $[v_3]=[e_3]$. 
	Then for all $-\theta_0 \leq \theta \leq \theta_0$ we have
	$$
	a = \left| \langle e_3,h\rangle \right| \iff a = \left| \widehat{z}(\theta) + \widehat{c_3}\langle e_3, h\rangle \right| \qquad (h\perp \{e_1, e_2\}, \|h\| = \rho(\theta)).
	$$
	Since $\# \Ta = \infty$, there exists at least one pair $(\theta, h)$ which solves both equations above. Note that $\langle e_3,h\rangle \neq 0$, and that $(\theta, \lambda h)$ also solves the first, hence the second, equation for all $\lambda\in\T$. By a simple geometric consideration one sees that this can happen only if $\widehat{c_3} = 1$. Therefore $[\widehat{v_3}]=[e_3]$, which further implies the contradiction $T = \{[v_1],[v_2],[e_3]\}$.
	Hence we obtain $[e_3]\notin T$.

	Therefore, neither $z$ nor $\widehat{z}$ is the constant zero function. 
	In particular, since their images are contained in (possibly degenerate) ellipses, they have at most two zeros. We distinguish two cases.
	
	\emph{Case 1. When for every $\theta\in[-\theta_0,\theta_0]$ we have $\#\left(\mathcal{A}_{{\theta}}\cap\{[v_3]\}^{\langle \alpha\rangle}\right) \leq 1$.}
	Define the set
	$$
	F := \left\{ \theta\in[-\theta_0,\theta_0] \colon \#\left(\mathcal{A}_{{\theta}}\cap\{[v_3]\}^{\langle \alpha\rangle}\right) = 1 \right\}.
	$$
	Since $\#\Ta = \infty$, we obtain $\# F =\infty$. 
	By (iii) of Lemma \ref{infinite-element}, we infer that 
	\begin{equation}\label{eq:c3}
	\big||z(\theta)|-a\big| = c_3\rho(\theta)
	\;\;\; (\theta\in F).
	\end{equation}
	
	We claim that \eqref{eq:c3} implies that $|z(\theta)|$ is constant on $[-\theta_0,\theta_0]$. In order to see this, we take the square of both sides in \eqref{eq:c3}, rearrange the equation, and take squares again:
	\begin{equation}\label{eq:c32}
	\left( |z(\theta)|^2+a^2-c_3^2\rho(\theta)^2 \right)^2 = \left(2a |z(\theta)|\right)^2 \;\;\; (\theta\in F).
	\end{equation}
	Notice that $\rho(\theta)^2$ and $|z(\theta)|^2$ are complex linear combinations of $\cos^2\theta, \sin^2\theta$ and $\cos\theta\sin\theta$. Hence they, and in particular the right-hand side of \eqref{eq:c32}, are complex linear combinations of $1, \cos(2\theta)$ and $\sin(2\theta)$. The left-hand side of \eqref{eq:c32} can be written in the form 
	\begin{align*}
	&(\mathfrak{a}+\mathfrak{b}\cos(2\theta)+\mathfrak{c}\sin(2\theta))^2 \\
	&= \mathfrak{a}^2 + \mathfrak{b}^2\cos^2(2\theta)+\mathfrak{c}^2\sin^2(2\theta) + 2 \mathfrak{a}\mathfrak{b}\cos(2\theta) + 2 \mathfrak{a}\mathfrak{c}\sin(2\theta) + 2 \mathfrak{b}\mathfrak{c}\cos(2\theta)\sin(2\theta)
	\end{align*}
	with some complex numbers $\mathfrak{a}, \mathfrak{b}, \mathfrak{c}$. Note that this expression is a complex linear combination of $1, \cos(2\theta),\sin(2\theta), \cos(4\theta)$ and $\sin(4\theta)$. Since both sides of \eqref{eq:c32} are trigonometric polynomials and they coincide on the infinite set $F\subset\left[-\frac{\pi}{2},\frac{\pi}{2}\right]$, they must coincide on the whole real line. Hence the coefficients on both sides with respect to $1, \cos(2\theta),\sin(2\theta), \cos(4\theta)$ and $\sin(4\theta)$ have to be the same. Since it is zero for $\sin(4\theta)$, we obtain that $\mathfrak{b} = 0$ or $\mathfrak{c} = 0$. Assume we have $\mathfrak{b} = 0$, then the left-hand side of \eqref{eq:c32} is 
	$$
	\mathfrak{a}^2 + \mathfrak{c}^2\sin^2(2\theta) + 2 \mathfrak{a}\mathfrak{c}\sin(2\theta)
	= \mathfrak{a}^2 + \frac{\mathfrak{c}^2}{2} - \frac{\mathfrak{c}^2}{2}\cos(4\theta) + 2 \mathfrak{a}\mathfrak{c}\sin(2\theta).
	$$
	But since the coefficient of $\cos(4\theta)$ is also zero, we obtain that $\mathfrak{c} = 0$. Therefore $|z(\theta)|$ is indeed a (non-zero) constant function. Similarly, we obtain the same conclusion for the $\mathfrak{c} = 0$ case.
	Using this information in \eqref{eq:c3} we obtain that $\rho(\theta)$ is constant on $F$, hence on $[-\theta_0,\theta_0]$. Therefore we infer $c=d=\frac{1}{\sqrt{2}}$, which contradicts our assumption $c>d$, so the present case cannot happen.
	
	\emph{Case 2. When there exists a $\widetilde{\theta}\in[-\theta_0,\theta_0]$ such that $\#\left(\mathcal{A}_{\widetilde{\theta}}\cap\{[v_3]\}^{\langle \alpha\rangle}\right) = \infty$ holds.}
	We claim that there is a non-degenerate interval $J\subseteq [-\theta_0,\theta_0]$ such that (i) from Lemma \ref{infinite-element} holds for all $\theta\in J$. If $\widetilde{\theta}$ satisfies (i), then this is clear from the continuity of $z$ and $\rho$. 
	Suppose $\widetilde{\theta}\neq 0$ and it satisfies (ii), namely, $z(\widetilde{\theta}) = 0$ and $\rho(\widetilde{\theta}) = \frac{a}{c_3}$.
	In this case if we move $\theta$ a little bit away from $\widetilde{\theta}$ but closer towards $0$, then (as $c>d>0$) both $|z(\theta)|$ and $\rho(\theta)$ increase continuously. Hence we get the desired interval.
	Finally, assume that $\widetilde{\theta} = 0$ and it satisfies (ii), namely, $z(0) = 0$ and $\rho(0) = \frac{a}{c_3}$. Consequently, $c_1=0$, and since $z$ is not constant zero, $c_2\neq 0$. We only have to observe that $|z(\theta)| =  \frac{a}{d} |c_2 \sin\theta|$ is differentiable from the right at $0$, and that this half-sided derivative is $|c_2| \frac{a}{d} > 0$. Since $\rho'(0) = 0$, we get the same conclusion by elementary calculus.

	Now, for all $\theta\in J$ there exists a non-degenerate arc $C_\theta$ in the complex plane such that
	\begin{align*}
	a = \left| z(\theta) + {c_3}\langle e_3, h\rangle \right| &\iff a = \left| \widehat{z}(\theta) + \widehat{c_3}\langle e_3, h\rangle \right| \\
	&\iff \langle e_3, h\rangle \in C_\theta \hspace{2cm} (h\perp \{e_1, e_2\}, \|h\| = \rho(\theta)).
	\end{align*}
	As the radii of the circles containing the arcs $z(\theta) + {c_3} C_\theta$ and $\widehat{z}(\theta) + \widehat{c_3}C_\theta$ are both equal to $a$, we obtain $\widehat{c_3} = c_3$. A consideration of their centres also gives $z(\theta) = \widehat{z}(\theta)$ $(\theta\in J)$. Since both $z$ and $\widehat{z}$ are trigonometric polynomials, their coincidence on the interval $J$ implies $c_1=\widehat{c_1}$, $c_2=\widehat{c_2}$, and hence $[\widehat{v_3}]=[v_3]$. So this second case cannot happen either. The proof is done.
\end{proof}

As it turns out the above lemma fails in three dimensions. The aim of the next section is to explore what can be said about highly-$\alpha$-symmetric sets in that case.

At this point the reader has the option to proceed with Section \ref{sec:proof} and read the proof of Theorem \ref{thm:main} in the case when $\dim H \geq 4$.


\section{The structure of highly-$\alpha$-symmetric sets in the three-dimensional case}\label{sec:4}

We start with a simple statement.

\begin{lemma}\label{closed}
The $\alpha$-set $\Sa$ of any subset $S\subset P(H)$ is closed. In particular, every highly-$\alpha$-symmetric set $T$ is compact, hence they contain at least one element that is not an isolated point of $T$.
\end{lemma}


The proof is straightforward, hence it is omitted. 
We now prove the three-dimensional version of Lemma \ref{circle4}.

\begin{lemma}\label{circle3}
Using the notation and assumptions of Lemma \ref{collin-alpha}, suppose that $\dim H = 3$ and that $e_3\perp\{e_1,e_2\}$ is a unit vector. Then we have the following possibilities:
\begin{itemize}
	\item[(i)] if either $\frac{c}{\sqrt{1+c^2}} \geq a > d$, or $(a,c,d) = \left(\frac{1}{\sqrt{3}}, \sqrt{\frac{2}{3}}, \frac{1}{\sqrt{3}}\right)$, or $(a,c,d) = \left(\frac{1}{\sqrt{3}}, \frac{1}{\sqrt{2}}, \frac{1}{\sqrt{2}}\right)$, then we have 
	\begin{align}\label{eq:double-circle}
		\Szaa &=\Big\{[ce_1 + \lambda de_2] \colon \lambda\in \T\Big\}  \bigsqcup \left\{\left[\sqrt{1-\frac{a^2}{1-\frac{a^2}{c^2}}}e_2 + \lambda \frac{a}{\sqrt{1-\frac{a^2}{c^2}}} e_3\right] \colon \lambda\in \T\right\},
	\end{align}
	\item[(ii)] otherwise we have
	\begin{equation}\label{eq:circle3}
		\Szaa = \{[ce_1 + \lambda de_2] \colon \lambda\in \T\}.
	\end{equation}
\end{itemize}
\end{lemma}

Note that if $(a,c,d) = \left(\frac{1}{\sqrt{3}}, \sqrt{\frac{2}{3}}, \frac{1}{\sqrt{3}}\right)$, then \eqref{eq:double-circle} becomes
\begin{align}\label{eq:circle-circle}
	\Szaa =\left\{\left[\sqrt{\frac{2}{3}}e_1 + \lambda \sqrt{\frac{1}{3}}e_2\right] \colon \lambda\in \T\right\}  \bigsqcup \left\{\left[\sqrt{\frac{1}{3}}e_2 + \lambda \sqrt{\frac{2}{3}} e_3\right] \colon \lambda\in \T\right\},
\end{align}
and if $\frac{c}{\sqrt{1+c^2}} = a$, then \eqref{eq:double-circle} is
\begin{align*}
		\Szaa = \Big\{[ce_1 + \lambda de_2] \colon \lambda\in \T\Big\}  \bigsqcup \left\{\left[e_3\right]\right\}.
\end{align*}
In particular, if $(a,c,d) = \left(\frac{1}{\sqrt{3}}, \frac{1}{\sqrt{2}}, \frac{1}{\sqrt{2}}\right)$, then \eqref{eq:double-circle} takes the form
\begin{align}\label{eq:circle-point}
	\Szaa = \left\{\left[\sqrt{\frac{1}{2}}e_1 + \lambda \sqrt{\frac{1}{2}}e_2\right] \colon \lambda\in \T\right\}  \bigsqcup \left\{\left[e_3\right]\right\}.
\end{align}

\begin{proof}[Proof of Lemma \ref{circle3}]
	As in the proof of Lemma \ref{circle4}, we set
	\begin{equation*}
		\mathcal{C} := \left\{\left[\frac{a}{c}e_1 + \lambda \sqrt{1-\frac{a^2}{c^2}} e_3\right]\colon \lambda\in\T \right\}.
	\end{equation*}
	Consider a line $[v] = [c_1e_1+c_2e_2+c_3e_3]$ with $c_1\geq 0$, $c_2, c_3\in\C$ and $c_1^2+|c_2|^2+|c_3|^2 = 1$.
	We obtain that $[v] \in \mathcal{C}^{\langle\alpha\rangle}$ if and only if 
	\begin{equation*}
	\left| c_1 \frac{a}{c} + c_3\overline{\lambda} \sqrt{1-\frac{a^2}{c^2}} \right| = a \qquad \left(\lambda\in\T\right).
	\end{equation*}
	As $c>a$, this is equivalent to
	\begin{itemize}
	\item either $c_1=c$, $|c_2| = d$ and $c_3=0$,
	\item or $c_1=0$, $|c_2| = \sqrt{1 - \frac{a^2}{1-\frac{a^2}{c^2}}}$ and $|c_3| = \frac{a}{\sqrt{1-\frac{a^2}{c^2}}}$.
	\end{itemize}
	Note that $a^2 > {1-\frac{a^2}{c^2}}$ holds if and only if $a > \frac{c}{\sqrt{1+c^2}}$.
	Therefore we obtain the following two possibilities:
	\begin{itemize}
	\item if $a > \frac{c}{\sqrt{1+c^2}}$, then
		$$\mathcal{C}^{\langle\alpha\rangle} = \{[ce_1 + \lambda de_2] \colon \lambda\in \T\},$$
	\item if $a \leq \frac{c}{\sqrt{1+c^2}}$, then
		$$\mathcal{C}^{\langle\alpha\rangle} = \{[ce_1 + \lambda de_2] \colon \lambda\in \T\} \bigsqcup \left\{\left[\sqrt{1 - \frac{a^2}{1-\frac{a^2}{c^2}}}e_2 + \lambda \frac{a}{\sqrt{1-\frac{a^2}{c^2}}} e_3\right] \colon \lambda\in \T\right\}.$$
	\end{itemize}
	In particular, this completes the case when $a > d$, since then $\Sza = \mathcal{C}^{\langle\alpha\rangle}$ holds. 
	
	In what follows, we shall handle the cases $a < d$ and $a = d$ separately.
	Set
	$$\mathcal{D} := \left\{\left[\frac{a}{d}e_2 + \lambda \sqrt{1-\frac{a^2}{d^2}} e_3\right]\colon \lambda\in\T\right\}.$$
	Suppose that $a < d$. 
	Then similarly as for $\mathcal{C}^{\langle\alpha\rangle}$ (where $c>a$ was automatic), we obtain the following:
	\begin{itemize}
	\item if $a > \frac{d}{\sqrt{1+d^2}}$, then
		$$\mathcal{D}^{\langle\alpha\rangle} = \{[ce_1 + \lambda de_2] \colon \lambda\in \T\},$$
	\item if $a \leq \frac{d}{\sqrt{1+d^2}}$, then
		$$\mathcal{D}^{\langle\alpha\rangle} = \{[ce_1 + \lambda de_2] \colon \lambda\in \T\} \bigsqcup \left\{\left[\sqrt{1 - \frac{a^2}{1-\frac{a^2}{d^2}}}e_1 + \lambda \frac{a}{\sqrt{1-\frac{a^2}{d^2}}} e_3\right] \colon \lambda\in \T\right\}.$$
	\end{itemize}
	Recall that $\Sza = \mathcal{C}\cup\mathcal{D}$. Hence we observe that 
	$$\{[ce_1 + \lambda de_2] \colon \lambda\in \T\} \subseteq \Szaa = \mathcal{C}^{\langle\alpha\rangle}\cap\mathcal{D}^{\langle\alpha\rangle} \subseteq \{[ce_1 + \lambda de_2] \colon \lambda\in \T\}\sqcup\{[e_3]\}.$$
	Therefore, after some easy calculations we obtain the following, which completes the $a < d$ case: 
	\begin{itemize}
	\item if $(a,c,d) = \left(\frac{1}{\sqrt{3}}, \frac{1}{\sqrt{2}}, \frac{1}{\sqrt{2}}\right)$, then we have \eqref{eq:double-circle} and \eqref{eq:circle-point},
	\item otherwise, we have \eqref{eq:circle3}.
	\end{itemize}
	
	Finally, let us assume that $a = d$. In this case $\mathcal{D} = \{[e_2]\}$, hence 
	$$
	S_0^{\langle\!\langle\alpha\rangle\!\rangle} = \mathcal{C}^{\langle\alpha\rangle}\cap\{[e_2]\}^{\langle\alpha\rangle} = \mathcal{C}^{\langle\alpha\rangle}\cap\{[de_2+x]\colon x\perp e_2, \|x\|=c\}.
	$$
	If we also have $a=d > \frac{c}{\sqrt{1+c^2}}$, then this clearly gives \eqref{eq:circle3}.
	Otherwise, 
	\begin{align*}
	S_0^{\langle\!\langle\alpha\rangle\!\rangle} = & \left\{[ce_1 + \lambda de_2], \left[\sqrt{1 - \frac{d^2}{1-\frac{d^2}{c^2}}}e_2 + \lambda \frac{d}{\sqrt{1-\frac{d^2}{c^2}}} e_3\right] \colon \lambda\in \T\right\} \\ 
	&\bigcap \Big\{[de_2+x]\colon x\perp e_2, \|x\|=c\Big\}.
	\end{align*}
	This gives \eqref{eq:circle3}, unless $\sqrt{1 - \frac{d^2}{1-\frac{d^2}{c^2}}} = d$, which happens if and only if $(a,c,d) = \left(\frac{1}{\sqrt{3}}, \sqrt{\frac{2}{3}}, \frac{1}{\sqrt{3}}\right)$. In this latter case we obtain \eqref{eq:double-circle} and \eqref{eq:circle-circle}, which completes the proof.
\end{proof}

Assume that the assumption of (ii) in Lemma \ref{circle3} holds. Then by Lemma \ref{collin-alpha} the circle in \eqref{eq:circle3} is highly-$\alpha$-symmetric. In the next lemma we investigate the other case.

\begin{lemma}\label{circle-implied-3}
	Assume that $\dim H = 3$ and that the assumptions of (i) in Lemma \ref{circle3} hold. Then the set $\Szaa$ in \eqref{eq:double-circle} is not highly-$\alpha$-symmetric.
\end{lemma}

\begin{proof}
Our strategy is to find four lines $[u_1],[u_2],[u_3]\in\Szaa$ and $[w]\in P(H)$ such that 
\begin{equation}\label{eq:M3}
	[w]\in \{[u_1], [u_2], [u_3]\}^{\langle \alpha\rangle} \setminus \Sza
\end{equation}
which, by (iii)--(iv) of Lemma \ref{basic}, will prove our statement.
Let $0 < t < \frac{\pi}{2}$ and consider the unit vector
\[
w := \frac{a}{c}\cos t \cdot e_1 + \frac{a}{c} \sin t \cdot e_2 + \sqrt{1-\frac{a^2}{c^2}} e_3,
\]
Note that $\langle w, e_j\rangle \neq 0$ $(j=1,2,3)$, thus by Lemma \ref{collin-alpha} we have $[w]\notin \Sza$. 
Using elementary calculus, it is easy to see that for small enough $t>0$ we have
\begin{equation}\label{eq:Mineq}
0 < c \left(\frac{a}{c}\cos t\right) - d \left(\frac{a}{c} \sin t\right) < a <  c \left(\frac{a}{c}\cos t\right) + d \left(\frac{a}{c} \sin t\right).
\end{equation}
Hence there exists a number $\lambda\in \T\setminus\{1, -1\}$ with 
$$
[u_1] := [ce_1 +\lambda d e_2],\; [u_2] := [ce_1 +\overline{\lambda} d e_2]\in \{[w]\}^{\langle \alpha\rangle}\cap\Szaa.
$$ 
In a similar way, we obtain the following for small enough $t>0$:
\[
0 < a - \sqrt{1-\frac{a^2}{1-\frac{a^2}{c^2}}}  \left(\frac{a}{c}\sin t\right) 
\leq a \leq 
a + \sqrt{1-\frac{a^2}{1-\frac{a^2}{c^2}}}  \left(\frac{a}{c}\sin t\right).
\]
However, unlike in \eqref{eq:Mineq}, here we have equations if $\frac{c}{\sqrt{1+c^2}} = a$.
Therefore, we conclude the existence of a number $\mu \in \T$ such that
\[
[u_3]:= \left[ \sqrt{1-\frac{a^2}{1-\frac{a^2}{c^2}}} e_2 + \mu \frac{a}{\sqrt{1-\frac{a^2}{c^2}}} e_3\right]\in \{[w]\}^{\langle \alpha\rangle}\cap\Szaa. 
\]
The relation \eqref{eq:M3} follows and the proof is complete.
\end{proof}

	We continue with the analogue of Lemma \ref{circle-char-4} in three dimensions. 
	It basically says that highly-$\alpha$-symmetric sets are exactly the circles with certain diameters. The lemma also implies some estimations for the diameter.
	
\begin{lemma}\label{circle-char-3}
Assume that $\dim H = 3$. Then for any set $T\subset P(H)$ and orthonormal system $\{e_1,e_2\}\subset H$ the following hold:
	\begin{itemize}
	\item[(i)] If $T$ is highly-$\alpha$-symmetric, then it is a circle.
	\item[(ii)] If $a\neq \frac{1}{\sqrt{3}}$, $c \geq d > a$, $c^2+d^2=1$, then the circle $\{[ce_1 + \lambda de_2] \colon \lambda\in \T\}$ is highly-$\alpha$-symmetric.
	\item[(iii)] If $a = \frac{1}{\sqrt{3}}$, $c > d > a$, $c^2+d^2=1$, then the circle $\{[ce_1 + \lambda de_2] \colon \lambda\in \T\}$ is highly-$\alpha$-symmetric.
	\item[(iv)] If $0 < d < \min\left\{a,\sqrt{\frac{1-2a^2}{1-a^2}}\right\}$, $c^2+d^2=1$, then the circle $\{[ce_1 + \lambda de_2] \colon \lambda\in \T\}$ is not highly-$\alpha$-symmetric.
	\end{itemize}
\end{lemma}

\begin{proof}
	Parts (ii)--(iv) easily follow from Lemmas \ref{collin-alpha} and \ref{circle3}. 
For (iv) we additionally note that $d < \sqrt{\frac{1-2a^2}{1-a^2}}$ implies $a < \frac{c}{\sqrt{1+c^2}}$.
	
	In order to prove (i), assume that $T$ is highly-$\alpha$-symmetric. 
	Suppose that there are three different elements $[v_1], [v_2], [v_3]\in T$ which sit on the same projective line. Then by Lemmas \ref{3}, \ref{circle3} and \ref{circle-implied-3}, the set $T = \{[v_1], [v_2], [v_3]\}^{\langle\!\langle \alpha\rangle\!\rangle}$ is a circle.
 
	From now on, we shall assume that no three different elements of $T$ are collinear. 
	Our aim is to obtain a contradiction. 
By Lemma \ref{closed}, we may take a line $[v_1]\in T$ that is not isolated in $T$.	
Take another line $[v_2]\in T\setminus\{[v_1]\}$. 
	They can be written as $[v_1] = [ce_1+ ide_2]$ and $[v_2] =[ce_1- ide_2]$ with some orthonormal system $\{e_1, e_2\}\subset H$ and real numbers $c\geq d>0$, $c^2+d^2 = 1$. 
Let $e_3\perp\{e_1,e_2\}$ be a unit vector.
In what follows, we use the same symbols as in the proof of Lemma \ref{circle-char-4}. 	
Namely, we consider two arbitrary lines $[v_3], [\widehat{v_3}]\in T\setminus\{[v_1], [v_2]\}$ which may be written as $[v_3] = [c_1e_1+c_2e_2+c_3e_3]$ and $[\widehat{v_3}] = [\widehat{c_1}e_1+\widehat{c_2}e_2+\widehat{c_3}e_3]$, where $c_1, c_2, \widehat{c_1}, \widehat{c_2}\in\C$, $c_3> 0$, $\widehat{c_3}> 0$, $|c_1|^2+|c_2|^2+c_3^2=|\widehat{c_1}|^2+|\widehat{c_2}|^2+\widehat{c_3}^2=1$.
	By (iv) of Lemma \ref{basic}, we have $\{[v_1], [v_2], [v_3]\}^{\langle \alpha\rangle} = \Ta = \{[v_1], [v_2], [\widehat{v_3}]\}^{\langle \alpha\rangle}$.
	By Lemma \ref{shape} this implies 
	$$\mathcal{A}_\theta \cap \{[v_3]\}^{\langle \alpha\rangle} = \mathcal{A}_\theta \cap \{[\widehat{v_3}]\}^{\langle \alpha\rangle}$$
	for all $-\theta_0 \leq \theta \leq \theta_0$, where
	$$
	\mathcal{A}_{\theta}:= \left\{ \left[\frac{a}{c}\cos\theta \cdot e_1 + \frac{a}{d}\sin\theta \cdot e_2 + \overline{\lambda}\rho(\theta)e_3\right] \colon
\lambda\in\T \right\}.
	$$
	In particular, observe that the cardinality 
	$$c(\theta) := \#\left(\mathcal{A}_\theta \cap \{[v_3]\}^{\langle \alpha\rangle}\right)$$ 
	does not depend on the specific choice of $[v_3]\in T\setminus\{[v_1], [v_2]\}$.
	We have the following for all $-\theta_0 \leq \theta \leq \theta_0$:
	\begin{equation}\label{eq:A-theta-part-3}
	a = \left| z(\theta) + c_3\rho(\theta)\lambda \right| \iff a = \left| \widehat{z}(\theta) + \widehat{c_3}\rho(\theta)\lambda \right| \quad (\lambda\in\T).
	\end{equation}
	
Exactly the same argument as in the proof of Lemma \ref{circle-char-4} right after \eqref{eq:A-theta-part} shows that $[e_3]\notin T$, hence neither $z$ nor $\widehat{z}$ is the constant zero function. We distinguish two cases.
		
	\emph{Case 1. When for all $\theta\in[-\theta_0,\theta_0]$ we have $c(\theta) < \infty$.} As can be seen by a simple geometric consideration, in this case for all $-\theta_0 < \theta <\theta_0$ (which implies $\rho(\theta)>0$) both equations of \eqref{eq:A-theta-part-3} have at most two solutions. In particular, there is no solution if $z(\theta) = 0$. Let $F$ be the set of those $\theta\in (-\theta_0, \theta_0)$ for which there is at least one solution $\lambda$. Note that $\# F =\infty$, as $\#\Ta = \infty$. For all $\theta\in F$ let $\lambda_1(\theta)$ and $\lambda_2(\theta)$ denote the two solutions, which might coincide for some $\theta$. By elementary geometry, one sees that $z(\theta)$ and $\widehat{z}(\theta)$ are real linearly dependent for all $\theta\in F$. 
Indeed, we can easily see the following: if $\lambda_1(\theta)=\lambda_2(\theta)$, then both $\{z(\theta), \lambda_1(\theta)\}$ and $\{\widehat{z}(\theta), \lambda_1(\theta)\}$ are real linearly dependent; if $\lambda_1(\theta)\neq \lambda_2(\theta)$, then both $z(\theta)$ and $\widehat{z}(\theta)$ are orthogonal to $\lambda_1(\theta)-\lambda_2(\theta)$ in the complex plane.
Hence for all $\theta\in F$
	\begin{align}\label{eq:z-zh-conj}
		0 = \frac{1}{a^2}\Im\left(z(\theta)\overline{\widehat{z}(\theta)}\right) = \frac{\Im\left(c_1\overline{\widehat{c_1}}\right)}{c^2}\cos^2\theta + \frac{\Im\left(c_2\overline{\widehat{c_2}}\right)}{d^2}\sin^2\theta + \frac{\Im\left(c_1\overline{\widehat{c_2}} + c_2\overline{\widehat{c_1}}\right)}{cd}\sin\theta\cos\theta.
	\end{align}
	Note that a trigonometric polynomial has infinitely many zeros on a compact interval if and only if it is the constant zero function on $\R$. Therefore, the right-hand side of \eqref{eq:z-zh-conj} is zero for all real $\theta$. By substituting $\theta = 0, \frac{\pi}{2}, \arccos c$, we obtain that each of the following is a real linearly dependent system in $\C$:
	\begin{equation}\label{eq:real-lin-dep}
		\left\{c_1, \widehat{c_1}\right\}, \;\; \left\{c_2, \widehat{c_2}\right\}, \;\; \left\{c_1 + c_2, \widehat{c_1}+\widehat{c_2}\right\}. 
	\end{equation}
		 
	Assume for a moment that $c_1$ and $c_2$ are real linearly independent complex numbers. 
Then \eqref{eq:real-lin-dep} implies $\widehat{c_1} = qc_1$ and $\widehat{c_2} = qc_2$ with some $0\neq q\in\R$. Notice that this forces $[\widehat{v_3}]$ to lie on the projective line spanned by $[e_3]$ and $\left[\frac{c_1}{|c_1|^2+|c_2|^2}e_1+\frac{c_2}{|c_1|^2+|c_2|^2}e_2\right]$, hence the contradiction $\#T \leq 4$ follows. Therefore we conclude that $c_1$ and $c_2$ are real linearly dependent, 
hence
	\begin{align*}
		\left|\langle v_3, v_1\rangle\right| = \left|cc_1 - idc_2 \right| = \sqrt{c^2\cdot|c_1|^2+d^2\cdot|c_2|^2} \leq c.
	\end{align*}
	Since $[v_3]\in T\setminus\{[v_1],[v_2]\}$ was arbitrary, we obtain that 
	$$\inf\{\measuredangle([v_1],[u]) \colon [u]\in T\setminus\{[v_1]\}\} > 0.$$
	This contradicts our assumption that $[v_1]$ is not an isolated point of $T$, so this case cannot happen.
	
	\emph{Case 2. When there exists a $\widetilde{\theta}\in[-\theta_0,\theta_0]$ such that $c(\widetilde{\theta}) = \infty$ holds.} In this case $\rho(\widetilde{\theta})>0$ and both equations in \eqref{eq:A-theta-part-3} are solved by infinitely many, hence all $\lambda\in\T$. Therefore, 
	\begin{equation}\label{eq:vmi}
	\widehat{z}(\widetilde{\theta}) = \widehat{c_1}\frac{a}{c}\cos\widetilde{\theta}+\widehat{c_2}\frac{a}{d}\sin\widetilde{\theta} = 0, \;\;\; z(\widetilde{\theta}) = c_1\frac{a}{c}\cos\widetilde{\theta}+c_2\frac{a}{d}\sin\widetilde{\theta} = 0
	\end{equation}
	and $c_3 = \widehat{c_3} = \frac{a}{\rho(\theta)}$.
	After some easy calculation we infer from \eqref{eq:vmi} that $(0,0) \neq (\widehat{c_1},\widehat{c_2}) = \mu(c_1,c_2)$ holds with some $\mu\in\T$. Therefore, 
	$[\widehat{v_3}]$ must lie on the projective line spanned by $[e_3]$ and $\left[\frac{c_1}{|c_1|^2+|c_2|^2}e_1+\frac{c_2}{|c_1|^2+|c_2|^2}e_2\right]$. However, since $[\widehat{v_3}] \in T\setminus\{[v_1], [v_2]\}$ was arbitrary, this implies the contradiction $\#T \leq 4$. So this case cannot happen either, the proof is done.
\end{proof}


\section{Proof of the main theorem}\label{sec:proof}

This section is devoted to the final step of the proof of our main result.

\begin{proof}[Proof of Theorem \ref{thm:main}]
	Let $M$ be an arbitrary two-dimensional subspace of $H$. In what follows we shall prove that there exists another two-dimensional subspace $N$ such that $\phi$ maps $P(M)$ onto $P(N)$. Then a straightforward application of Theorem \ref{thm:2d} gives that the restriction $\phi|_{P(M)}$ preserves every quantum angle, which in turn completes the proof. 
	
	Fix 
	$$
	c_0 := \left\{\begin{matrix}
		\sqrt{\frac{1}{2}}, & \text{if}\;\; \dim H \geq 4 \;\; \text{or} \;\; a\neq\sqrt{\frac{1}{3}} \\
		\sqrt{\frac{7}{12}}, & \text{if}\;\; \dim H = 3 \;\; \text{and} \;\; a=\sqrt{\frac{1}{3}}
	\end{matrix}\right.
	$$ 
	and $d_0 := \sqrt{1-c_0^2}$. By Lemmas \ref{circle-char-4} and \ref{circle-char-3}, every circle of the form 
	$$
	C([e_1],[e_2]) := \{[c_0e_1 + \lambda d_0e_2] \colon \lambda\in \T\},
	$$
	where $\{e_1,e_2\}$ is an orthonormal system, is highly-$\alpha$-symmetric. We obviously have
	$$
	P(M) = \bigcup \big\{ C([e_1],[e_2]) \colon \{e_1,e_2\} \;\; \text{is an orthonormal basis of}\;\; M\big\}.
	$$
	It is apparent from Definition \ref{defi} and the properties of $\phi$, that $\phi$ and $\phi^{-1}$ map highly-$\alpha$-symmetric sets onto highly-$\alpha$-symmetric sets. In particular, $\phi\left(P(M)\right)$ is a union of circles of the form $D([e_1],[e_2]) := \phi\big(C([e_1],[e_2])\big)$.
	
	Observe that if $\#\left(C([e_1],[e_2])\cap C([f_1],[f_2])\right) \geq 2$ holds for two orthonormal bases $\{e_1,e_2\}$ and $\{f_1,f_2\}$ of $M$, then $D([e_1],[e_2])$ and $D([f_1],[f_2])$ are contained in the same projective line.
	Indeed, there exist two different lines $[u_1], [u_2] \in P(M)$ such that $\{[u_1],[u_2]\} \subseteq C([e_1],[e_2])\cap C([f_1],[f_2])$. Set $[v_1] := \phi([u_1])$ and $[v_2] := \phi([u_2])$. Since $[v_1], [v_2]\in D([e_1],[e_2])\cap D([f_1],[f_2])$, we conclude $D([e_1],[e_2]) \cup D([f_1],[f_2]) \subseteq P([v_1, v_2])$.
Note that $P([v_1,v_2])$ is equal to the projective line generated by $\phi([c_0e_1 + d_0e_2])$ and $\phi([c_0e_1 - d_0e_2])$.
	
	From here we distinguish between two cases.
	
	\emph{Case 1. When $\dim H \geq 4$ or $a\neq\sqrt{\frac{1}{3}}$ holds.}
	Then $c_0=d_0=\frac{1}{\sqrt{2}}$, and it is rather straightforward to see from the Bloch representation that $\#\left(C([e_1],[e_2])\cap C([f_1],[f_2])\right) \geq 2$ holds for all pairs of orthonormal bases $\{e_1,e_2\}$ and $\{f_1,f_2\}$ of $M$. Indeed, the Bloch representations of these circles are great circles on $\mathbb{S}^2$. However, let us give here a more direct proof of the inequality $\#\left(C([e_1],[e_2])\cap C([f_1],[f_2])\right) \geq 2$. If $\{[e_1],[e_2]\} = \{[f_1],[f_2]\}$, then this is obvious, so from now on we assume otherwise. There exist numbers $\mathfrak{a}, \mathfrak{b}> 0$, $\mathfrak{a}^2+\mathfrak{b}^2 = 1$, $\mu\in\T$ such that $f_1$ may be assumed to have the form $\mathfrak{a}e_1+\mu\mathfrak{b}e_2$. Consequently, $f_2$ may be assumed to have the form $\mathfrak{b}e_1-\mu\mathfrak{a}e_2$. Then
	$$
	\left[\frac{1}{\sqrt{2}}e_1+i\mu\frac{1}{\sqrt{2}}e_2\right] = \left[\frac{\mathfrak{a}-i\mathfrak{b}}{\sqrt{2}}e_1+\mu\frac{\mathfrak{b}+i\mathfrak{a}}{\sqrt{2}}e_2\right] = \left[\frac{1}{\sqrt{2}}f_1-i\frac{1}{\sqrt{2}}f_2\right] \in C([e_1],[e_2])\cap C([f_1],[f_2]),
	$$
	and similarly
	$$
	\left[\frac{1}{\sqrt{2}}e_1-i\mu\frac{1}{\sqrt{2}}e_2\right] = \left[\frac{\mathfrak{a}+i\mathfrak{b}}{\sqrt{2}}e_1+\mu\frac{\mathfrak{b}-i\mathfrak{a}}{\sqrt{2}}e_2\right] = \left[\frac{1}{\sqrt{2}}f_1+i\frac{1}{\sqrt{2}}f_2\right] \in C([e_1], [e_2])\cap C([f_1], [f_2]).
	$$
	
	Fix an orthonormal basis $\{e_1,e_2\}$ of $M$. 
	We obtain 
	\begin{align*}
	\phi(P(M)) &= \phi\left(\bigcup \big\{ C([f_1],[f_2]) \colon \{f_1,f_2\} \;\; \text{is an orthonormal basis of}\;\; M\big\}\right)\\
	&= \bigcup \big\{ D([f_1],[f_2]) \colon \{f_1,f_2\} \;\; \text{is an orthonormal basis of}\;\; M\big\}\\
	&\subseteq P(N),
	\end{align*}
	where $P(N)$ is the projective line generated by $\phi([c_0e_1 + d_0e_2])$ and $\phi([c_0e_1 - d_0e_2])$.
	However, by the very same reasons, the inverse $\phi^{-1}$ maps $P(N)$ into some projective line $P(L)$. Since we have $P(M)\subseteq\phi^{-1}(P(N))$, we infer $P(M)=\phi^{-1}(P(N))$, which in turn completes the proof of this case.
	
	\emph{Case 2. When $\dim H = 3$ and $a = \sqrt{\frac{1}{3}}$ are satisfied.} 
	Then $c_0 = \sqrt{\frac{7}{12}}$ and $d_0 = \sqrt{\frac{5}{12}}$.
	One easily sees that it suffices to show the following: for any two orthonormal bases $\{e_1,e_2\}$ and $\{f_1,f_2\}$ of $M$, there exists a third orthonormal basis $\{g_1,g_2\}$ of $M$ such that 
	\begin{equation}\label{eq:g1g2}
	\#\left(C([e_1],[e_2])\cap C([g_1],[g_2])\right) \geq 2, \;\;\; \#\left(C([g_1],[g_2])\cap C([f_1],[f_2])\right) \geq 2.
	\end{equation}
	Again, one way to verify this is by utilising the Bloch representation, however, let us show it directly here. 
	If $[e_1] = [f_1]$ and $[e_2] = [f_2]$, then this is obvious, so from now on we assume otherwise. 
	Then there are numbers $0 \leq \mathfrak{a} < 1, 0 < \mathfrak{b}\leq 1$, $\mathfrak{a}^2+\mathfrak{b}^2 = 1$, $\mu\in\T$ such that $f_1$ and $f_2$ may be assumed to have the forms $\mathfrak{a}e_1+\mu\mathfrak{b}e_2$ and $\mathfrak{b}e_1-\mu\mathfrak{a}e_2$, respectively.
	Note that
	\begin{align*}
	C([f_1],[f_2]) & = \left\{\left[\sqrt{\frac{7}{12}}\left(\mathfrak{a}e_1+\mu\mathfrak{b}e_2\right) + \lambda \sqrt{\frac{5}{12}}\left(\mathfrak{b}e_1-\mu\mathfrak{a}e_2\right)\right] \colon \lambda\in \T\right\} \\
	& = \left\{\left[\left(\sqrt{\frac{7}{12}}\mathfrak{a}+\lambda \sqrt{\frac{5}{12}}\mathfrak{b}\right)e_1+\mu\left(\sqrt{\frac{7}{12}}\mathfrak{b} - \lambda \sqrt{\frac{5}{12}}\mathfrak{a}\right)e_2\right] \colon \lambda\in \T\right\}.
	\end{align*}
	As $\mathfrak{b} > 0$, we obtain that $\#\left(C([e_1],[e_2])\cap C([f_1],[f_2])\right) \geq 2$ is satisfied if and only if there exists a $\lambda\in\T\setminus\{-1,1\}$ such that $\left|\sqrt{\frac{7}{12}}\mathfrak{a}+\lambda \sqrt{\frac{5}{12}}\mathfrak{b}\right| = \sqrt{\frac{7}{12}}$. The latter equation is equivalent to $\left|\mathfrak{a}+\lambda \sqrt{\frac{5}{7}}\mathfrak{b}\right| = 1$. 
	Since $\left|\mathfrak{a}-\sqrt{\frac{5}{7}}\mathfrak{b}\right| \leq \max\left\{\mathfrak{a}, \sqrt{\frac{5}{7}}\right\} <1$, the inequality $\#\left(C([e_1],[e_2])\cap C([f_1],[f_2])\right) \geq 2$ holds if and only if $\mathfrak{a}+\sqrt{\frac{5}{7}}\mathfrak{b} > 1$. A simple calculation gives that this is further equivalent to $\frac{1}{6} < \mathfrak{a} < 1$.
	
	Note that $\mathfrak{a} = |\langle e_1,f_1\rangle|$. Therefore, if we have $\frac{1}{6} < |\langle e_1,f_1\rangle| < 1$, then \eqref{eq:g1g2} holds with $g_1=e_1$ and $g_2=e_2$. On the other hand, if $0\leq |\langle e_1,f_1\rangle| = \mathfrak{a} \leq \frac{1}{6}$, then choose $g_1 := \frac{1}{\sqrt{2}}e_1+\mu\frac{1}{\sqrt{2}}e_2$ and $g_2 := \frac{1}{\sqrt{2}}e_1-\mu\frac{1}{\sqrt{2}}e_2$. We have $\left|\langle g_1,e_1\rangle\right| = \frac{1}{\sqrt{2}} > \frac{1}{6}$ and
	$$
	\left|\langle g_1,f_1\rangle\right| = \frac{1}{\sqrt{2}}\mathfrak{a} + \frac{1}{\sqrt{2}}\mathfrak{b} \geq \frac{1}{\sqrt{2}}\frac{\sqrt{35}}{6}  > \frac{1}{6}.
	$$
	This completes the proof.
\end{proof}

We close our paper with mentioning that even though $H$ was assumed to be a Hilbert space, our method clearly works for general complex inner product spaces as well. In that case, the only change we have to make in the statement of Theorem \ref{thm:main} is to replace ``unitary or an antiunitary operator'' with ``bijective linear or conjugate-linear isometry'', since the former term is usually used only for Hilbert spaces.


\end{document}